\documentclass[preprint, 5p, twocolumn]{elsarticle}

\usepackage[pdftex]{hyperref}
\hypersetup{
	pdfauthor={Szabolcs Szentpéteri, Balázs Csanád Csáji},
	pdftitle={Non-Asymptotic State-Space Identification of Closed-Loop\\ Stochastic Linear Systems using Instrumental Variables},
    colorlinks = true,
    citecolor = {blue},
}
\usepackage{amsmath,amssymb,amsfonts}
\usepackage[english]{babel}
\usepackage{hyphenat}
\usepackage{amsthm}
\usepackage{nicefrac}
\usepackage{algorithm, algorithmic}

\journal{}
\def \newtext{}

\DeclareMathOperator{\Tr}{tr}

\newcommand{\defeq}{\doteq}
\newcommand{\frob}{_\mathrm{F}}
\newcommand{\norm}[1]{\left\lVert#1\right\rVert}

\def \CE{\mathcal{E}}

\usepackage{mathtools}

\newcommand{\tr}{^\mathrm{T}}  %

\newcommand{\R}{\mathbb{R}}

\newtheorem{theorem}{Theorem}
\newtheorem{lemma}{Lemma}
\newtheorem{definition}{Definition}
\newtheorem*{assumption*}{\assumptionnumber}
\providecommand{\assumptionnumber}{}
\makeatletter
\newenvironment{assumption}[1]
{%
	\renewcommand{\assumptionnumber}{A#1}%
	\begin{assumption*}%
		\protected@edef\@currentlabel{A#1}%
	}
	{%
	\end{assumption*}
}
\makeatother

\begin{document}

\begin{frontmatter}

\title{\LARGE\vspace*{-2mm} Non-Asymptotic State-Space Identification of Closed-Loop\\ Stochastic Linear Systems using Instrumental Variables}
\author[1]{Szabolcs Szentpéteri}
\author[1,2]{\hspace*{20mm}Balázs Csanád Csáji}
\address[1]{Institute for Computer Science and Control (SZTAKI), E\"otv\"os Lor\'and Research Network (ELKH), Budapest, Hungary\vspace{1mm}}
\address[2]{Department of Probability Theory and Statistics, Institute of Mathematics, E\"otv\"os Lor\'and University (ELTE), Budapest, Hungary\vspace*{-4mm}}

\begin{abstract}
The paper suggests a generalization of the Sign-Perturbed Sums (SPS) finite sample system identification method for the identification of closed-loop observable stochastic linear systems in state-space form. The solution builds on the theory of matrix-variate regression and instrumental variable methods to construct distribution\hyp free confidence regions for the state-space matrices. Both direct and indirect identification are studied, and the exactness as well as the strong consistency of the construction are proved. Furthermore, a new, computationally efficient ellipsoidal outer-approximation algorithm for the confidence regions is proposed. The new construction results in a semidefinite optimization problem which has an order-of-magnitude smaller number of constraints, as if one applied the ellipsoidal outer-approximation after vectorization. The effectiveness of the approach is also demonstrated empirically via a series of numerical experiments.

\end{abstract}

\begin{keyword}
closed-loop identification, distribution-free methods, non-asymptotic guarantees, instrumental variables
\end{keyword}

\end{frontmatter}

\section{Introduction}
\vspace*{-1mm}
Estimating a {\em mathematical model} from observations of a {\em dynamical system} is a fundamental problem across several fields, from system identification to signal processing and machine learning. Standard estimation techniques, 
such as the prediction error approach or the generalized method of moments, typically provide {\em point estimates} associated with {\em only asymptotically guaranteed} confidence sets \cite{Ljung1999}.

In practical applications there is often {\em limited statistical knowledge} about the noises and uncertainties affecting the system and a {\em limited number of measurements}. Furthermore, in many situations, data can only be gathered under {\em feedback control}. This is especially the case, for example, for economy and biology related applications. If our problem involves strong safety or stability requirements, having guaranteed confidence regions are strongly desirable. 
{\newtext Typical examples,} where guaranteed region estimates can be essential are robust and adaptive control \cite{Bertsekas2007, kumar2015stochastic, khalil1996robust}.

One of the first approaches to non-asymptotic system identification were \cite{weyer1999} and \cite{Campi2002}. In the past few years, due to a renewed interest in system identification methods with finite sample guarantees, non-asymptotic high probability bounds on the estimation error were investigated for linear state-space models, {\newtext see for example
\cite{faradonbeh2018, sarkar2019, jedra2019sample, oymak2019non, tsiamis2019finite, zheng2020non}.} The above results often use {\em strong statistical assumptions}, such as joint {\em Gaussianity}, and \cite{faradonbeh2018, sarkar2019, tsiamis2019finite} considered uncontrolled, while \cite{jedra2019sample} and \cite{zheng2020non} studied open-loop LTI systems. Non-asymptotic identification of a {\em closed-loop} LTI system with the REDAR algorithm was investigated in \cite{lee2020non}, still assuming {\em Gaussian} white noise type uncertainties. 

Typical examples of {\em distribution\hyp free} system identification algorithms with 
strong 
non-asymptotic guarantees are
the Leave\hyp out Sign\hyp dominant Correlation Regions (LSCR) \cite{Campi2005,Campi2009,Campi2010} and the Sign\hyp Perturbed Sums (SPS) \cite{Csaji2015, volpe2015sign, Csaji2015cdc} methods. 
SPS can construct {\em exact} confidence regions for (open-loop) {\em general linear systems} under mild statistical assumptions \cite{Csaji2012b}. In standard SPS, the confidence set is given by its {\em indicator} function, which can be evaluated at any parameter.
In \cite{kieffer2014guaranteed} a guaranteed characterization of SPS was developed using interval analysis{\newtext, while in \cite{Csaji2015}
an {\em ellipsoidal outer approximation} was given for FIR (finite impulse response) systems.}
The {\em strong consistency} and the asymptotic shape of SPS regions were studied in \cite{Weyer2017}.

Several generalizations of SPS were suggested, such as Data Perturbation (DP) methods \cite{kolumban2015perturbed} which can use other perturbations, not only sign changes; and UD-SPS that is able to detect undermodelling, in case of FIR systems \cite{care2021facing}.

The {\em closed-loop} applicability of SPS for 
general linear systems
was studied in \cite{Csaji2015cdc}, where the {\em exact coverage} of 
SPS regions
is shown, but they 
are only given by their indicators. 
An {\em instrumental variables} \cite{Soderstrom1989} based extension of SPS for linear regression %
was presented in \cite{volpe2015sign},
{\newtext making the SPS {\em outer ellipsoids} applicable to ARX (autoregressive with exogenous inputs) systems even under {\em feedback}.}

All of the SPS related papers above studied {\em scalar} models, where the system is SISO and it is given either with transfer functions or in a linear regression form. In many control problems the system is MIMO and it is given in a {\em state-space} form which even includes feedback control, for which the SPS variants above are not directly applicable. 

One could apply the scalar approach of \cite{volpe2015sign} to a state-space model via vectorization, i.e., after reformulating it as a (potentially huge) linear regression problem  \cite{giacomo2022state}.
This, however, needs additional assumptions on the noise vectors (i.e., that their distributions are symmetric w.r.t.\ each coordinate axis and even their components are independent). Further, computing the resulting ellipsoidal outer approximation leads to large semidefinite programming problems.

Instead of doing so, we extend SPS to state-space models by applying {\em matrix-variate regression} which leads to a more compact approach with relaxed assumptions on the noises (the components of the noise vectors can be dependent and only their joint distributions should be symmetric about zero, full axial symmetry is not needed). An alternative {\em semidefinite program} (SDP) is suggested, as well, to compute outer ellipsoids, which has an {\em order-of-magnitude} smaller number of constraints than the one based on \cite{volpe2015sign}.

The main {\em contributions} of the paper are as follows:
\begin{enumerate}
	\item A generalization of SPS is proposed to build distribution\hyp free confidence sets for {\em matrix-variate regression} problems using {\em instrumental variables} (MIV-SPS). It is demonstrated on {\em closed-loop state-space} models.\vspace{1mm}
	\item The {\em exact coverage probabilities} (for any finite sample size) and the {\em strong consistency} of the MIV-SPS confidence regions are proved under mild assumptions.\vspace{1mm}
	\item A new, efficient {\em ellipsoidal outer approximation} algorithm is introduced for MIV-SPS, based on SDP.\vspace{1mm}
	\item The effectiveness of MIV-SPS is also validated experimentally and it is compared to (scalar) IV-SPS and the confidence ellipsoids of the asymptotic theory.
\end{enumerate}
The paper is organized as follows. In Section \ref{problemsetting} the problem setting and its linear and matrix-variate regression formulations as well as our main assumptions are introduced. Section \ref{sec:lsandasym} gives a summary of the instrumental variable (IV) method and its asymptotic theory. In Section \ref{sec:spsindicator} the MIV-SPS algorithm is presented, while Section \ref{sec:theory} states the theoretical guarantees of the method. Section \ref{sec:eoa} introduces the ellipsoidal outer-approximation algorithm. The simulation experiments and comparisons are presented in Section \ref{sec:experiments}. Finally, Section \ref{conclusion} concludes the paper.

\vspace*{-1mm}
\section{Problem Setting}\label{problemsetting}
\vspace*{-1mm}
\label{sec:problemsetting}
This section introduces the closed-loop stochastic linear state-space model, presents the linear and matrix-variate regression reformulations and the main assumptions.

\vspace*{-1mm}
\subsection{Stochastic Linear State-Space Model}
Consider the observable linear state-space (LSS) model
\vspace{-1mm}
\begin{equation}\label{equ:lss_st}
	x_{k+1}\, =\, Ax_k\, +\, Bu_k\, +\, w_k,
\vspace{-1mm}
\end{equation}
for $k=0, 1, \dots, n-1$, where $x_k$ is a $d_x$-dimensional, $u_k$ is a $d_u$-dimensional, and $w_k$ is a $d_x$-dimensional real random vector; $A$, $B$ are {\em unknown} real matrices
to be estimated. The system can operate in {\em closed-loop} with feedback rule
\vspace{-1mm}
\begin{equation}\label{equ:lss_action}
	u_k\, =\, Fx_k \,+\, Gr_k,
\vspace{-1mm}	
\end{equation}
where $r_k$ is an $d_r$-dimensional real random vector, which can represent a reference signal, a setpoint or noise affecting the controller; while $F$ and $G$ are 
real matrices. 

In this work, we will concentrate on {\em region estimation}, therefore, we aim at constructing a {\em confidence region} which contains matrices $A$ and $B$ with a {\em user-chosen} probability. 

\vspace*{-1mm}
\subsection{Linear Regression Formulation}\label{sec:linregression}
\label{subsec:linregform}
The LSS dynamics \eqref{equ:lss_st} can be written as
\begin{equation}
	\begin{aligned}
		\begin{bmatrix}
			x_{k+1, 1}\\
			x_{k+1, 2}\\
			\vdots \\
			x_{k+1, d_x}
		\end{bmatrix}=&\begin{bmatrix}
			a\tr_1\\
			a\tr_2\\
			\vdots \\
			a\tr_{d_x}
		\end{bmatrix}\begin{bmatrix}
			x_{k, 1}\\
			x_{k, 2}\\
			\vdots \\
			x_{k, d_x}
		\end{bmatrix} + \\
		&\begin{bmatrix}
			b\tr_1\\
			b\tr_2\\
			\vdots \\
			b\tr_{d_u}
		\end{bmatrix}\begin{bmatrix}
			u_{k, 1}\\
			u_{k, 2}\\
			\vdots \\
			u_{k, d_u}
		\end{bmatrix}+\begin{bmatrix}
			w_{k, 1}\\
			w_{k, 2}\\
			\vdots \\
			w_{k, d_x}
		\end{bmatrix}\!\!,
	\end{aligned}
\end{equation}
where $\{a_i\tr\}$ and $\{b_i\tr\}$ are the rows of $A$ and $B$, respectively.
Then, the system dynamics from time $1$ until time $n$ can be written in a {\em linear regression} form as
\begin{equation}
	\label{sys:linreg}
	y \,=\, \Xi \,\theta^* +\, w,
\end{equation}
where $\theta^* \defeq (a_1\tr, \dots,a_{d_x}\tr, b_1\tr, \dots, b_{d_u}\tr)\tr\!\!\!,$\, and
\begin{equation}
	y \defeq \begin{bmatrix}
		x_{1,1} \\
		x_{1,2} \\
		\vdots \\
		x_{1,d_x} \\
		x_{2,1} \\		
		\vdots \\
		x_{n,d_x} \\
	\end{bmatrix}\!\!,\quad
	\Xi \defeq \begin{bmatrix}
		\xi_{0,1}\tr\\
		\xi_{0,2}\tr\\
		\vdots \\
		\xi_{0,d_x}\tr\\
		\xi_{1,1}\tr\\		
		\vdots \\
		\xi_{n-1,d_x}\tr\\
	\end{bmatrix}\!\!,\quad
	w \defeq \begin{bmatrix}
		w_{0,1}\\
		w_{0,2}\\
		\vdots \\
		w_{0,d_x}\\
		w_{1,1}\\	
		\vdots \\
		w_{n-1,d_x}\\
	\end{bmatrix}\!\!,
\end{equation}
and the regressors $\{\xi_{k, i}\}$ are defined as follows
\begin{equation}
	\label{eq:xik}
	\xi_{k, i}\tr \defeq (\overbrace{0_x\tr,\hdots,0_x\tr,\underbrace{x_k\tr}_{i \text{th}}, 0_x\tr}^{\text{dim:}\, d_x^2}, \overbrace{0_u\tr ,\hdots, 0_u\tr,\underbrace{u_k\tr}_{i \text{th}},0_u\tr}^{\text{dim:}\, d_xd_u}),
\end{equation}
where $0_x$ and $0_u$  are $d_x$ and $d_u$ dimensional zero vectors.

The advantage of the linear regression formulation of \eqref{sys:linreg} is that the previous results for scalar systems, e.g., the IV-SPS method \cite{Csaji2015cdc}, can be directly applied. Note that in this case each marginal distribution of the noise vector should satisfy the noise assumptions of the SPS, in order to allow the application of scalar SPS approaches.

\vspace*{-1mm}
\subsection{Matrix-Variate Regression Formulation}\label{sec:matrixregression}
In this paper, we argue that a more natural way to handle (even closed-loop) linear state-space models is to apply a matrix-variate regression formulation. As we will see, this allows relaxed assumption on the noises and leads to computationally more efficient constructions.

Without vectorization, state-space model \eqref{equ:lss_st} can be reformulated as a {\em matrix-variate regression} problem, i.e.,
\begin{equation}
	\label{sys-matvarreg}
	Y = \Phi\Theta^*\hspace{-0.5mm} + W,
\end{equation}
where, unlike in \eqref{sys:linreg}, the output, $Y$, the true parameter, $\Theta^*$, and the noise, $W$, are all matrices, as well. That is
\begin{equation*}
	Y \defeq \begin{bmatrix}
		x_1\tr \\
		x_2\tr \\
		\vdots \\
		x_n\tr
	\end{bmatrix}\!\!, \hspace{2mm}
	\Phi \defeq \begin{bmatrix}
		\varphi_0\tr \\
		\varphi_1\tr \\
		\vdots \\
		\varphi_{n-1}\tr
	\end{bmatrix}\!\!,\hspace{2mm}
	\Theta^* \defeq \begin{bmatrix}
		A\tr \\
		B\tr
	\end{bmatrix}\!\!,\hspace{2mm}
	W \defeq \begin{bmatrix}
		w_0\tr \\
		w_1\tr \\
		\vdots \\
		w_{n-1}\tr
	\end{bmatrix}\!\!,
\end{equation*}
here the regressors $\varphi_k$ are defined as $\varphi_k \defeq (x_k\tr, u_k\tr)\tr\!\!\!,$\, %
which are $d \defeq d_x + d_u$ dimensional vectors.

{\newtext We will also use the notations $Y_n$, $\Phi_n$, 
and $W_n$, in cases when the dependence on the sample size $n$ is crucial.}

\vspace*{-1mm}
\subsection{Indirect identification}
In Sections \ref{sec:linregression} and \ref{sec:matrixregression} the regression problems were formulated using the {\em direct} system identification approach, where the values $\{x_k\}$ and $\{u_k\}$ are used during estimation. The regression problems can also be formulated using the {\em indirect} approach \cite{forssell1999closed}, where the control matrices $F$ and $G$ are {\em known}, which is often the case in an adaptive control setting. 
By expanding the state transition equation we can reformulate\vspace{-1mm}
\eqref{equ:lss_st} and \eqref{equ:lss_action} as follows
\begin{equation}
	\begin{aligned}
		x_{k+1} &= Ax_k + B(Fx_k + Gr_k) + w_k \\
 		&= (A+BF)x_k + BGr_k + w_k \\
        &= Cx_k + Dr_k + w_k.
	\end{aligned}
    \vspace{-1mm}	
\end{equation}
In the indirect identification scheme, using the values $\{x_k\}$ and $\{r_k\}$, the $C$ and $D$ matrices are estimated, from which $A$ and $B$ can be computed (under mild conditions on the control matrices $F$ and $G$). In this case
\eqref{sys-matvarreg} should of course be modified accordingly, for example, the rows of the regressor matrix $\Phi_{\text{id}}$ are given by $\bar{\varphi}_k \defeq (x_k\tr, r_k\tr)\tr$, and the true parameter $\Theta^*_{\text{id}}$ should contain matrices $C$ and $D$. 

The linear regression formulation can be obtained similarly to the approach of Section \ref{subsec:linregform}, hence omitted.

\vspace*{-1mm}
\subsection{Core Assumptions}
\label{sec:assumptions}

Henceforth, we will directly study the matrix-variate regression problem \eqref{sys-matvarreg}, irrespectively whether it came from the direct or the indirect version of the original identification problem. Our main assumptions will be as follows:
	\begin{assumption}{1}\label{assu:noise}
	The row vectors $\{w_k\}$\hspace{-0.3mm} of the noise matrix $W$\hspace{-0.7mm} are independent, and they are distributed symmetrically about zero {\em(}but, they can have different distributions{\em)}, i.e., for all $k$, random vectors $w_k$  and $-w_k$ have the same distribution.
	\end{assumption}
	
	This assumption is very mild, even milder than the one we get from standard SPS \cite{Csaji2015, volpe2015sign, Csaji2015cdc, Csaji2012b, Weyer2017}, if we try to apply it to a vectorized version of \eqref{sys-matvarreg}. 
    For example, such a vectorized approach would require that for all $k$ and $i$,
    \vspace{-1mm}
    \begin{equation}
    \label{axial-sym}
    \begin{aligned}
    (w_{k,1}, \dots, w_{k,i-1}, w_{k,i}, w_{k,i+1}, \dots, w_{k,d_x}) \overset{\scriptscriptstyle d}{=}\\
    (w_{k,1}, \dots, w_{k,i-1}, -w_{k,i}, w_{k,i+1}, \dots, w_{k,d_x}),
    \end{aligned}
    \vspace{-1mm}
    \end{equation}
    {\newtext where ``$\overset{\scriptscriptstyle d}{=}$'' denotes equality in distribution. This is strictly stronger than $w_k \overset{\scriptscriptstyle d}{=} -w_k$, as the latter is implied by \eqref{axial-sym}, but not the other way around.
    Moreover, \ref{assu:noise} also allows the components of the noise vector $w_k$ to be {\em correlated}.}

	\begin{assumption}{2}\label{assu:iv1}
    We are given a random matrix $\Psi \defeq (\psi_0, \dots, \psi_{n-1})\tr$ $\in \mathbb{R}^{n \times d}$, for which matrices $\Psi$\hspace{-0.3mm} and $W$\hspace{-1mm} are independent.
	\end{assumption}	

{\newtext The rows of matrix $\Psi$ are called the {\em instrumental variables}. %
In a typical situation, they are filtered past inputs using an auxiliary model \cite{Ljung1999}.
The motivation for using instrumental variables will be discussed in Section \ref{sec:intuitiv_idea}.}
	\begin{assumption}{3}\label{assu:iv2}
	Matrix $\Psi\tr\Phi$ is full rank almost surely.
	\end{assumption}

Observe that from assumption $\ref{assu:iv2}$ it follows that matrix $\Psi\tr\Psi$ is also full rank (invertible) almost surely. We also introduce the notations 
$V_n \defeq \nicefrac{1}{n}\, \Psi\tr\Phi,$ and $P_n \defeq \nicefrac{1}{n}\, \Psi\tr\Psi.$

Note, as well, that instrumental variable based identification methods typically also assume that the instrumental variables, $\{\psi_k\}$, and the states, $\{x_k\}$, are correlated. 
In the strict sense, this assumption is not needed to construct exact confidence regions for the true parameter matrix $\Theta^*$, but an asymptotic correlation assumption \ref{assu:inv_v} is crucial to prove consistency.

\vspace{-1mm}
\section{Instrumental Variable Methods}\label{sec:lsandasym}
\vspace*{-1mm}
Now, we briefly overview identification with {\em instrumental variable} (IV) methods and its asymptotic theory \cite{Ljung1999}.

\vspace*{-1mm}
\subsection{Instrumental Variable Estimate (IVE)}
By defining the prediction of a particular $\Theta$ matrix by $\hat{Y} \defeq \Phi\Theta$, we compute the {\em prediction errors} (residuals) as
$	\CE(\Theta) \defeq Y - \hat{Y} = Y -\Phi\Theta$.
IVE is obtained by solving
\begin{equation}\label{normalequ}
	\Psi\tr \CE(\Theta) = \Psi\tr(Y-\Phi\Theta) = 0,
\end{equation}
then, assuming \ref{assu:iv2}, the IVE can be calculated as
\begin{equation}\label{iv-estimate}
	\widehat{\Theta}^{\scriptscriptstyle \text{IV}} = (\Psi\tr\Phi)^{-1}\Psi\tr Y,
\end{equation}
which is usually a {\em biased} estimator of $\Theta^*$, under \ref{assu:noise}-\ref{assu:iv2}, unless we assume that $\Phi$ and $W$ are independent.
\vspace*{-1mm}
\subsection{Limiting Distribution of IVE}
Now, assume that $Y_n$ is $\mathbb{R}^n$-valued, i.e., $\{x_k\}$ is a scalar process. Then, the ``true'' parameter is a constant vector, $\theta^* \in \mathbb{R}^{d_{\theta}}$\!\!, and the IVE is a random vector given by \eqref{iv-estimate}.

If $\{w_k\}$ are i.i.d., zero mean, and each has variance $\sigma^2 \in (0, \infty)$, then the IVE is {\em asymptotically Gaussian}:\vspace{-1mm}
\begin{equation}
	\label{eq:iv-asymp-normal}
	\sqrt{n}\,(\hat{\theta}^{\scriptscriptstyle \text{IV}}_n - \theta^*) \xrightarrow{\enskip d \enskip}\mathcal{N}(0, \sigma^2R^{-1}),
\end{equation}
as $n \rightarrow \infty$, with 
$R \defeq \lim_{n \rightarrow \infty} R_n$, where $\{R_n\}$ are\vspace{-1.5mm}
\begin{equation}
	R_n\, \defeq\, \left[\frac{1}{n}\Psi_n\tr \Phi_n\right]\tr\left[\frac{1}{n}\Psi_n\tr \Psi_n\right]^{-1}\left[\frac{1}{n}\Phi_n\tr \Psi_n\right]\hspace{-4mm},
\end{equation}
assuming that the limit exists and is positive definite \cite{Ljung1999},
where we (explicitly) emphasized the dependences on $n$.

In case $\{x_k\}$ is vector-valued, we can reformulated it in a linear regression form \eqref{sys:linreg}. Then, the (vectorized) IVE is also asymptotically Gaussian, with appropriate modification of the instrumental variables, similarly to \eqref{eq:xik}.	

\vspace*{-1mm}
\subsection{Asymptotic Confidence Regions of IVE}
Assuming $\widehat{\theta}^{\scriptscriptstyle \text{IV}}_n$ is (vector-valued and) asymptotically normal, that is if 
\eqref{eq:iv-asymp-normal} holds, then we also have \cite{Ljung1999} \vspace{-1mm}
\begin{equation}
	\frac{n}{\sigma^2}(\hat{\theta}^{\scriptscriptstyle \text{IV}}_n - \theta^*)\tr R\hspace{0.5mm} (\hat{\theta}^{\scriptscriptstyle \text{IV}}_n - \theta^*)\xrightarrow{\enskip d \enskip} \chi^2(d_{\theta}),
\end{equation}
as $n \to \infty$, where $\chi^2(d_{\theta})$ is the chi-square distribution with $d_{\theta} = \mathrm{dim}(\theta^*)$ degrees of freedom. Matrix $R$ and variance $\sigma^2$ are unknown, but they can be estimated. Matrix $R_n$ is {\newtext an empirical estimate} of $R$ and an estimate of $\sigma^2$ is\vspace{-1mm}
\begin{equation}
	\hat{\sigma}_n^2 \, \defeq\, \frac{1}{n-d_{\theta}}\hspace{0.5mm}\lVert\hspace{0.3mm} \CE(\hat{\theta}^{\scriptscriptstyle \text{IV}}_n) \hspace{0.3mm}\rVert_2^2.
\end{equation}
Then, a $p$-level {\em confidence ellipsoid} can be defined as
\begin{equation}
		{\tilde{\Upsilon}}_{n,p} \,\defeq \left\{\theta \in \mathbb{R}^{d_{\theta}}:(\theta - \hat{\theta}^{\scriptscriptstyle \text{IV}}_n)\tr R_n (\theta - \hat{\theta}^{\scriptscriptstyle \text{IV}}_n) \leq \frac{\mu\hat{\sigma}_n^2}{n}\right\}\!,\!\!
\end{equation}
where $p = F_{\chi^2(d_{\theta})}(\mu)$ is the target confidence probability and $F_{\chi^2(d_{\theta})}$ is the cumulative distribution function of the $\chi^2$ distribution. Then, we have $\mathbb{P}(\theta^* \in {\tilde{\Upsilon}}_{n,p}) \approx p$. Note that these regions only have {\em asymptotic} guarantees. 

\vspace*{-1mm}
\section{Matrix-Variate Generalization of IV-SPS}
\label{sec:spsindicator}
In this section, we introduce
a {\em matrix-variate} generalization of the SPS method using instrumental variables. Our construction can be efficiently applied to solve closed-loop state-space identification problems, formulated as \eqref{sys-matvarreg}.

The scalar variant of IV-SPS for linear regression problems was originally suggested in \cite{volpe2015sign}, which can be used for state-space models after suitable vectorization \cite{giacomo2022state}. 

One of the advantages of MIV-SPS w.r.t.\ standard SPS methods \cite{Csaji2015,volpe2015sign} is that the matrix-variate regression formulation allows perturbing the components of the noise vectors simultaneously, which lead to relaxed assumptions on the noises. Moreover, the matrix-variate approach also allows a computationally more efficient ellipsoidal outer approximation construction than the original approach.

\subsection{Intuitive Idea of the SPS Construction}\label{sec:intuitiv_idea}
Using the matrix-variate regression formulation \eqref{sys-matvarreg}, the {\em normal equation} \eqref{normalequ} can be reformulated as
\begin{equation}
	\Psi\tr\Phi(\Theta^* - \Theta) + \Psi\tr W = 0.
\end{equation}
Following the SPS principle \cite{Csaji2015,volpe2015sign}, we introduce $m-1$ sign-perturbed sums, to construct the region, as follows
\begin{equation}
	\begin{aligned}
		H_i(\Theta) \,\defeq \;\, &\Psi\tr\Lambda_i(Y-\Phi\Theta) = \\& \Psi\tr\Lambda_i\Phi(\Theta^* - \Theta) + \Psi\tr\Lambda_i W,
	\end{aligned}
\end{equation}
for $i = 1, \dots, m-1$, where $\Lambda_i$ is a 
diagonal matrix, $\Lambda_i \doteq \text{diag}(\alpha_{i,1}, \dots, \alpha_{i,n})$,
and $\{\alpha_{i,k}\}$ are i.i.d. Rademacher variables, i.e., $\mathbb{P}(\alpha_{i,k} = 1) = \mathbb{P}(\alpha_{i,k} = -1) = \nicefrac{1}{2}$. A reference sum is also introduced, without sign-perturbations,
\begin{equation}
	H_0(\Theta)\, \defeq\, \Psi\tr(Y-\Phi\Theta) = \Psi\tr\Phi(\Theta^* - \Theta) + \Psi\tr W.
\end{equation}

For the true parameter matrix, $\Theta = \Theta^*$, we have
\begin{equation}
	H_0(\Theta^*) = \Psi\tr W,\quad \text{and}\quad H_i(\Theta^*) = \Psi\tr \Lambda_i W.
\end{equation}
$H_0(\Theta)$ and $H_i(\Theta)$ can be compared using the Frobenius norm $\norm{\cdot}\frob$. According to assumption \ref{assu:noise}, the noise vector sequence $\{w_k\}$ is independent and symmetric about zero, thus $H_0(\Theta^*)$ and $H_i(\Theta^*)$ have the same distribution, hence the probability that a particular $\norm{H_\ell(\Theta^*)}\frob^2$ is ranked at a given position according to a strict total order of the values $\{\norm{H_j(\Theta^*)}\frob^2\}_{j=0}^{m-1}$ 
is the same for all $\ell$. Note that this property is not trivial, since $\{H_j(\Theta^*)\}_{j=0}^{m-1}$ are dependent.

When $\norm{\Theta^* - \Theta}\frob$ is ``sufficiently'' large, the inequality $\forall\, i \neq 0\!:\! \norm{H_0(\Theta)}\frob^2 > \norm{H_i(\Theta)}\frob^2$, will hold, thus $\norm{H_0(\Theta)}\frob^2$ will eventually be the largest of the $m$ functions. The core idea behind the SPS method is to construct the confidence region based on the rankings of %
$\{\norm{H_j(\Theta)}\frob^2\}_{j=0}^{m-1}$ and exclude those $\Theta$ parameter matrices for which the reference, $\norm{H_0(\Theta)}\frob^2$, is among the $q$ largest. As we will show, the so constructed confidence set has {\em exactly} probability $1-q/m$ of containing the true parameter matrix. In the final formal description, the functions $\{\norm{H_j(\Theta)}\frob^2\}_{j=0}^{m-1}$ will also be weighted with a suitable ``shaping'' matrix.

Instrumental variables (IVs) are introduced, in order to be able to handle closed-loop systems. In these cases, if $\Psi$ was simply replaced by $\Phi$, $\{\norm{H_j(\Theta^*)}\frob^2\}_{j=0}^{m-1}$ would not have the same distribution, because the regressor matrix is {\em not} independent of the noise terms. The IVs counteract these dependencies, thus they ensure the validity of the confidence set construction in closed-loop setups.

\subsection{Formal Confidence Region Construction}
The SPS algorithm consists of two main parts, an initialization and an indicator function. In the {\em initialization} part, see Algorithm \ref{alg:sps_init}, the input is the user defined confidence probability $p$. The algorithm computes the main parameters and generates the random objects needed for the construction of the confidence region. In the {\em indicator} part, see Algorithm \ref{alg:sps_indi}, the input of the algorithm is a particular parameter matrix $\Theta$, and the function evaluates whether $\Theta$ is included in the confidence region.
\begin{algorithm}[t]
	\caption{MIV-SPS: Initialization$\,(p)$}\label{alg:sps_init}
	\begin{algorithmic}[1]
		\STATE Given the (rational) confidence probability $p \in (0,1)$, set integers $m > q >0$ such that $p = 1 - q/m$.
		\STATE Calculate the outer product\vspace{-1mm}
		\[ P_n \,\defeq\, \tfrac{1}{n}\Psi\tr\Psi,\vspace{-2mm} \]
		and find the principal square root $P_n^{\nicefrac{1}{2}}$, such that\vspace{-1mm}
		\[P_n^{\nicefrac{1}{2}}P_n^{\nicefrac{1}{2}} = P_n.\vspace{-5mm}\]
		\STATE Generate $n(m-1)$ i.i.d random signs $\{\alpha_{i,k}\}$ with\vspace{-1mm}
		\[\mathbb{P}(\alpha_{i,k} = 1)\, =\, \mathbb{P}(\alpha_{i,k} = -1)\, =\, \tfrac{1}{2},\vspace{-1mm}\] 
		for $i \in \{1,\dots,m-1\}$, $k \in \{1,\dots,n\}$ and construct the following matrices containing these random signs\vspace{-1mm}
		\[
		\Lambda_i \, \defeq \, \begin{bmatrix}
			\alpha_{i,1} & & \\
			& \ddots & \\
			& & \alpha_{i,n}
		\end{bmatrix}\!.\vspace{-3mm}
		\]
		\STATE Generate a uniform random permutation $\pi$ of the set $\{0,\dots, m - 1\}$, where each of the $m!$ possible permutations has the same probability $1/(m!)$ to be selected.
	\end{algorithmic}
\end{algorithm}
\begin{algorithm}[t]
	\caption{MIV-SPS: Indicator$\,(\Theta)$}\label{alg:sps_indi}
	\begin{algorithmic}[1]
		\STATE For the given $\Theta$, compute the prediction errors\vspace{-1mm}
		\[ \CE(\Theta) \defeq Y -\Phi\Theta.\vspace{-7mm} \]
		\STATE Evaluate
		\begin{equation*}
		\begin{aligned}
		 	S_0(\Theta) &\defeq \tfrac{1}{n}P_n^{-\frac{1}{2}}\Psi\tr \CE(\Theta), \\[1mm]
		 S_i(\Theta) &\defeq \tfrac{1}{n}P_n^{-\frac{1}{2}}\Psi\tr\Lambda_i\CE(\Theta),
		\end{aligned}
		\end{equation*}
		for $i \in \{1,...,m-1\}$.
		\STATE Order scalars $\{\norm{S_i(\Theta)}\frob^2\}$ according to $\succ_{\pi}$, which is the standard ``$>$'' relation with random tie-breaking: 
        $\norm{S_k(\Theta)}\frob^2 \succ_{\pi} \norm{S_j(\Theta)}\frob^2$ if and only if $(\norm{S_k(\Theta)}\frob^2 > \norm{S_j(\Theta)}\frob^2) \lor
        (\norm{S_k(\Theta)}\frob^2 = \norm{S_j(\Theta)}\frob^2 \,\land\, \pi(k) > \pi(j))$.
		\STATE Compute the rank $\mathcal{R}(\Theta)$ of $\norm{S_0(\Theta)}\frob^2$ in the ordering:\vspace{-1mm}
		\begin{equation*}
			\mathcal{R}(\Theta)\, \defeq\, \left[\,1+\sum_{i=1}^{m-1}\mathbb{I}\left(\norm{S_0(\Theta)}\frob^2 \succ_{\pi} \norm{S_i(\Theta)}\frob^2\right)\right]\!.
			\vspace{-1mm}
		\end{equation*}
		\STATE Return 1 if $\mathcal{R}(\Theta) \leq m - q$, otherwise return 0.
	\end{algorithmic}
\end{algorithm}

By using the indicator function, the $p$-level SPS confidence region for $\Theta^*$ can be defined as follows
\begin{equation}
	\Upsilon_n \,\defeq\, \big\{\,\Theta \in \mathbb{R}^{d\times d_x}\text{ : Indicator}(\Theta) = 1\,\big\}.
\end{equation}
For the instrumental variable estimate $\widehat{\Theta}^{\scriptscriptstyle \text{IV}}_n$ it holds that $S_0(\widehat{\Theta}^{\scriptscriptstyle \text{IV}}_n) = 0$; therefore, the IV estimate is always in the SPS confidence region, assuming that it is non-empty. 
\vspace{-1mm}
\section{Theoretical Guarantees}\label{sec:theory}
\subsection{Exact Coverage Probability}
The confidence regions constructed by MIV-SPS have guaranteed coverage probabilities for the true parameter matrix, for any finite sample size. More precisely, we have:

\begin{theorem}\label{theorem1}
	Assuming \ref{assu:noise}-\ref{assu:iv2}, the confidence probability of the constructed confidence region is exactly $p$, that is,
	\vspace{-1mm}
\begin{equation}
	\mathbb{P}\big(\hspace{0.3mm}\Theta^* \in \Upsilon_n\hspace{0.3mm}\big) \, =\, 1-\frac{q}{m} \,=\, p.
\end{equation}
\end{theorem}

The proof of Theorem \ref{theorem1}, which is a generalization of the proof of \textit{Theorem 1} in \cite{volpe2015sign}, can be found in \ref{app:proof_theorem1}. 

Observe that the coverage probability is {\em exact}, the set is non-conservative.
The assumptions on the noises are rather mild: no specific (parametric) distributions are assumed, the noise sequence can even be non-stationary (i.e., each noise vector can have a different distribution), and there are no moment or density assumptions, either.

This result is more general than the one we would get by applying the corresponding result of \cite{volpe2015sign} to the vectorized version of \eqref{sys-matvarreg}, as the assumptions on the noises are milder.
\subsection{Strong Consistency}
Similarly to standard SPS, MIV-SPS is also uniformly {\em strongly consistent}, under the following 
assumptions:
\begin{assumption}{4}\label{assu:psd_p}
	There exists a positive definite matrix $P$ such that
	\begin{equation}
	\lim_{n\to\infty} \tfrac{1}{n}\, \Psi_n\tr\Psi_n = \lim_{n\to\infty} P_n = P\quad \text{{\em(}a.s.\hspace{0.5mm}{\em)}.}
	\end{equation}
\end{assumption}
\begin{assumption}{5}\label{assu:inv_v}
    There exists an invertible matrix $V$ such that
	\begin{equation}
    \lim_{n\to\infty} \tfrac{1}{n}\, \Psi_n\tr\Phi_n = \lim_{n\to\infty} V_n = V\quad \text{{\em(}a.s.\hspace{0.5mm}{\em)}.}
	\end{equation}
\end{assumption}
\begin{assumption}{6}\label{assu:r_grr}
The following growth rate restriction holds almost surely for the rows of the regression matrices:
\vspace{-1mm}
\begin{equation}
\sum_{k=1}^{\infty}\frac{\norm{\varphi_k}^4}{k^2} <\, \infty.
\end{equation}
\end{assumption}
\begin{assumption}{7}\label{assu:i_grr}
The following growth rate restriction holds almost surely for the rows of the instrumental variable matrices:
\vspace{-1mm}
\begin{equation}
\sum_{k=1}^{\infty}\frac{\norm{\psi_k}^4}{k^2} <\, \infty.
\end{equation}
\end{assumption}
\begin{assumption}{8}\label{assu:var_grr}
Finally, the variance growth of the rows of the noise matrices satisfy the following condition:
\vspace{-2mm}
\begin{equation}
\sum_{k=1}^{\infty}\frac{\Big(\mathbb{E}\big[\norm{w_k}^2\big]\Big)^2}{k^2} <\, \infty.
\end{equation}
\end{assumption}

Note that these assumptions are also rather mild. For example, if the regressors (or the instrumental variables) are bounded, then \ref{assu:r_grr} (or, respectively, \ref{assu:i_grr}) is satisfied. If the noise vectors $\{w_t\}$ are square integrable and i.i.d., then \ref{assu:var_grr} holds. However, these conditions allow unbounded regressors and instrumental variables as well as noise terms whose variances grow to infinity. We also note that \ref{assu:inv_v} ensures the asymptotic correlation of the instrumental variables and the regressors, which is crucial for consistency.

\begin{theorem}\label{theorem2}
	Assuming \ref{assu:noise}-\ref{assu:var_grr}, $\forall\,\varepsilon > 0,$ we have\vspace{-1mm}
\begin{equation}
\mathbb{P}\bigg(\bigcup_{k=1}^\infty \bigcap_{n=k}^\infty\! \big\{\,\Upsilon_n \subseteq \mathcal{B}_{\varepsilon}(\Theta^*) \big\} \bigg) =\, 1,
\vspace{-1mm}
\end{equation}
where $\mathcal{B}_{\varepsilon}(\Theta^*) \defeq \{\, \Theta \in \mathbb{R}^{d \times d_x}: \norm{\hspace{0.3mm}\Theta - \Theta^*}\frob \le \varepsilon\, \}$.
\end{theorem}
The proof can be found in \ref{app:proof_theorem2}. The theorem studies the event that, given $\varepsilon > 0$, after at most finite number of observations, the confidence regions $\{{\Upsilon}_n\}$ will remain included in the $\varepsilon$ ball around $\Theta^*$, as the sample size, $n$, increases. This {\em tail event} happens almost surely.

\section{Ellipsoidal Outer Approximation}\label{sec:eoa}
The MIV-SPS Indicator function evaluates whether a given parameter matrix $\Theta$ belongs to the confidence region, by comparing the $\{\norm{S_i(\Theta)}\frob^2\}_{i=1}^{m-1}$ functions with $\norm{S_0(\Theta)}\frob^2$. In many applications, we need a compact representation of the {\em whole} confidence set, to make it easier working with it; for example, in robust control.

In this section we introduce a generalization of the {\em ellipsoidal outer approximation} (EOA) given in \cite{Csaji2015} and \cite{volpe2015sign}. 
Our construction results in a semidefinite programming (convex) optimization problem with significantly less constraints than the one we would get by applying the approach of \cite{Csaji2015} and \cite{volpe2015sign} to the vectorized system.
\subsection{Ellipsoidal Outer Approximation}
Using the definition of the Frobenius norm, $\norm{S_0(\Theta)}\frob^2$ can be reformulated as follows
\vspace{-1mm}
\begin{align}
	&\norm{S_0(\Theta)}\frob^2 \nonumber \\ 
	&= \Tr\!\left(\left[\tfrac{1}{n}\Psi\tr(Y - \Phi\Theta)\right]\tr\! P_n^{-1} \left[\tfrac{1}{n}\Psi\tr(Y - \Phi\Theta)\right] \right) \nonumber \\ 
	&= \Tr\!\left(\left[\tfrac{1}{n}\Psi\tr\Phi(\Theta- \widehat{\Theta}^{\scriptscriptstyle \text{IV}}_n)\right]\tr\! P_n^{-1} \left[\tfrac{1}{n}\Psi\tr\Phi(\Theta- \widehat{\Theta}^{\scriptscriptstyle \text{IV}}_n)\right] \right) \nonumber \\
	&= \Tr\!\left((\Theta- \widehat{\Theta}^{\scriptscriptstyle \text{IV}}_n)\tr V_n\tr P_n^{-1} V_n(\Theta- \widehat{\Theta}^{\scriptscriptstyle \text{IV}}_n)\right)\nonumber \\
    &= \norm{P_n^{\nicefrac{-1}{2}} V_n(\Theta- \widehat{\Theta}^{\scriptscriptstyle \text{IV}}_n)}\frob^2\!.
\end{align}
To obtain an ellipsoid over-bound, consider the set given by the particular values of $\Theta$ for which $\norm{S_i(\Theta)}\frob^2 \geq \norm{S_0(\Theta)}\frob^2$ holds true for $q$ of the $\norm{S_i(\Theta)}\frob^2$, that is
\vspace{-1mm}
\begin{equation}
	\begin{aligned}
		&\Upsilon_n \subseteq\!
		\biggl\{\Theta \in \mathbb{R}^{d\times d_x}\!:\hspace{-3mm}  
        &\norm{P_n^{\nicefrac{-1}{2}} V_n(\Theta- \widehat{\Theta}^{\scriptscriptstyle \text{IV}}_n)}\frob^2
        \leq r(\Theta) \biggr\},
	\end{aligned}
	\vspace{-1mm}
\end{equation}
where $r(\Theta)$ is the $q$th largest value of $\{\norm{S_i(\Theta)}\frob^2\}_{i=1}^{m-1}$. 

{\newtext Our objective now is to determine an ellipsoidal over-bound with a parameter-independent radius, $r$, instead of the parameter-dependent $r(\Theta)$. This outer approximation is a guaranteed confidence region for any finite sample} with a compact representation given by $\widehat{\Theta}^{\scriptscriptstyle \text{IV}}_n$, $V_n$, $P_n$ and $r$.

\subsection{Convex Programming Formulation}\label{sec:cvxprog}
Comparing $\norm{S_0(\Theta)}\frob^2$ with a single $\norm{S_i(\Theta)}\frob^2$, we have
\vspace{-1mm}
\begin{equation}
	\begin{aligned}
		&\left\{\Theta : \norm{S_0(\Theta)}\frob^2 \leq \norm{S_i(\Theta)}\frob^2\right\} \subseteq \\[1mm] &\Big\{\Theta : \norm{S_0(\Theta)}\frob^2 \leq \max_{\Theta : \norm{S_0(\Theta)}\frob^2 \leq \norm{S_i(\Theta)}\frob^2} \norm{S_i(\Theta)}\frob^2\Big\}
	\end{aligned}
	\vspace{-1mm}
\end{equation}
Inequality  $\norm{S_0(\Theta)}\frob^2 \leq \norm{S_i(\Theta)}\frob^2$ can be reformulated as
\vspace{-1mm}
\begin{align}
	&\Tr\!\left((\Theta- \widehat{\Theta}^{\scriptscriptstyle \text{IV}}_n)\tr V_n\tr P_n^{-1} V_n(\Theta- \widehat{\Theta}^{\scriptscriptstyle \text{IV}}_n)\right) \leq \nonumber \\
	&\Tr\!\left(\left[\tfrac{1}{n}\Psi\tr\Lambda_i(Y-\Phi\Theta)\right]\tr P_n^{-1}\left[\tfrac{1}{n}\Psi\tr\Lambda_i(Y-\Phi\Theta)\right]\right) = \nonumber \\[0.5mm]
	&\Tr\!\left(\Theta\tr Q_i\tr P_n^{-1}Q_i\Theta - \Theta\tr Q_i\tr P_n^{-1}M_i \right. \nonumber \\[1mm]
	& \left. - M_i\tr P_n^{-1}Q_i\Theta + M_i\tr P_n^{-1}M_i\right)\!,
\end{align}
where matrices $Q_i$ and $M_i$ are defined as
\vspace{-1mm}
\begin{equation}
\begin{aligned}
	Q_i&\, \defeq\, \tfrac{1}{n}\Psi\tr\Lambda_i\Phi,\\[1mm]
	M_i& \,\defeq\, \tfrac{1}{n}\Psi\tr\Lambda_iY.
\end{aligned}
\vspace{-2mm}
\end{equation}
Observe that
\vspace{-1mm}
\begin{equation}
	\begin{aligned}
		&\max_{\Theta : \norm{S_0(\Theta)}\frob^2 \leq \norm{S_i(\Theta)}\frob^2} \norm{S_i(\Theta)}\frob^2 = \\ & \max_{\Theta : \norm{S_0(\Theta)}\frob^2 \leq \norm{S_i(\Theta)}\frob^2} \norm{S_0(\Theta)}\frob^2.
	\end{aligned}
\end{equation}
After defining matrix $Z \defeq P_n^{\nicefrac{-1}{2}}V_n(\Theta - \widehat{\Theta}^{\scriptscriptstyle \text{IV}}_n)$, the quantity 
$\max_{\Theta : \norm{S_0(\Theta)}\frob^2 \leq \norm{S_i(\Theta)}\frob^2} \norm{S_i(\Theta)}\frob^2$ 
can be computed by 
\vspace{-1mm}
\begin{equation}
	\begin{aligned}
		\max \quad & \norm{Z}^2\frob\\[1mm]
	  	\textrm{s.t.} \quad &\Tr(Z\tr A_iZ + Z\tr B_i + B_i\tr Z + C_i) \,\leq\, 0, \\
	\end{aligned}
\end{equation}
where $A_i$, $B_i$ and $C_i$ are defined as follows
\vspace{-1mm}
\begin{equation}
\begin{aligned}
	A_i \defeq \;& I-P_n^{\frac{1}{2}\mathrm{T}}V_n^{-\mathrm{T}}Q_i\tr P_n^{-1}Q_iV_n^{-1}P_n^{\frac{1}{2}}\\
	B_i \defeq \;& P_n^{\frac{1}{2}\mathrm{T}}V_n^{-\mathrm{T}}Q_i\tr P_n^{-1}(M_i-Q_i\widehat{\Theta}^{\scriptscriptstyle \text{IV}}_n)\\[1mm]
	C_i \defeq \;&-(\widehat{\Theta}^{\scriptscriptstyle \text{IV}}_n)\tr Q_i\tr P_n^{-1}Q_i\widehat{\Theta}^{\scriptscriptstyle \text{IV}}_n + (\widehat{\Theta}^{\scriptscriptstyle \text{IV}}_n)\tr Q_i\tr P_n^{-1}M_i \\ &+ M_i\tr P_n^{-1}Q_i\widehat{\Theta}^{\scriptscriptstyle \text{IV}}_n - M_i\tr P_n^{-1}M_i.
	\end{aligned}
\end{equation}
This program is {\em not} convex, however weak Lagrange duality holds, therefore an {\em upper bound} on the value of the above optimization problem can be obtained from the {\em dual} of the problem, for which an equivalent problem is
\vspace{-1mm}
\begin{equation}\label{cvxprog}
	\begin{aligned}
		\min \quad & \Tr(\Gamma - \lambda C_i)\\[1mm]
		\textrm{s.t.} \quad &\lambda \geq 0 \\[1mm]
		\quad &\begin{bmatrix}
			-I + \lambda A_i & \lambda B_i \\
			\lambda B_i\tr & \Gamma\\
		\end{bmatrix} \succeq 0,
		\\
	\end{aligned}
\end{equation}
where $\Gamma$ is a symmetric matrix and $\succeq 0$ denotes that the matrix is positive semidefinite. The derivation of this formulation can be found in \ref{app:cvxformulation}. The optimization problem \eqref{cvxprog} is now {\em convex} and can be solved using CVXPY \cite{diamond2016cvxpy}, e.g., with solver CVXOPT or MOSEK.

Let $\gamma_i^*$ be the solution of the program \eqref{cvxprog}. We have
\vspace{-1mm}
\begin{equation}
	\begin{aligned}
		&\left\{\Theta : \norm{S_0(\Theta)}\frob^2 \leq \norm{S_i(\Theta)}\frob^2\right\} \\
		&\subseteq \left\{\Theta : \norm{S_0(\Theta)}\frob^2 \leq \gamma_i^*\right\},
	\end{aligned}
    \vspace{-2mm}	
\end{equation}
therefore
\vspace{-1mm}
\begin{equation}
\label{upsilon_subset_hat_upsilon}
	\begin{aligned}
		&\Upsilon_n\, \subseteq\, \hat{\Upsilon}_n \defeq \biggl\{\Theta \in \mathbb{R}^{d\times d_x}: \\ 
        &\norm{P_n^{\nicefrac{-1}{2}} V_n(\Theta- \widehat{\Theta}^{\scriptscriptstyle \text{IV}}_n)}\frob^2
        \leq r\biggr\},
	\end{aligned}
	\vspace{-1mm}
\end{equation}
where $r$ is the $q$th largest value of $\{\gamma_i^*\}_{i=1}^{m-1}$. $\hat{\Upsilon}_n$ is an {\em ellipsoidal outer approximation} of $\Upsilon_n$, and from \eqref{upsilon_subset_hat_upsilon} we have
\vspace{-1mm}
\begin{equation}
	\mathbb{P}\big(\Theta^* \in \hat{\Upsilon}_n\big) \,\geq\, 1-\dfrac{q}{m} \,=\, p,
\end{equation}
for any $n$. The method is summarized by Algorithm \ref{alg:sps_eoa}.
\begin{algorithm}
	\caption{MIV-SPS: Ellipsoidal Outer Approximation}\label{alg:sps_eoa}
	\begin{algorithmic}[1]
		\STATE Compute the instrumental variable estimate,
		\vspace{-1mm}
		\[ \widehat{\Theta}^{\scriptscriptstyle \text{IV}}_n \defeq (\Psi\tr\Phi)^{-1}\Psi\tr Y. \]
		\vspace{-6mm}
		\STATE For $i \in \{1, \dots, m-1\}$ solve the optimization problem \eqref{cvxprog} and let $\gamma_i^*$ be the optimal value of the program.
		\STATE Let $r$ be the $q$th largest value of $\{\gamma_i^*\}_{i=1}^{m-1}$.
		\STATE The ellipsoidal outer approximation of the SPS confidence region is given by
		\vspace{-1mm}
		\begin{equation*}
			\begin{aligned}
				&\hat{\Upsilon}_n \defeq \left\{\Theta \in \mathbb{R}^{d\times d_x}: %
                \norm{P_n^{\nicefrac{-1}{2}} V_n(\Theta- \widehat{\Theta}^{\scriptscriptstyle \text{IV}}_n)}\frob^2
                \leq r\right\}\!.
			\end{aligned}
		\vspace*{-1mm}
		\end{equation*}
	\end{algorithmic}
\end{algorithm}

An important difference between the IV-SPS \cite{volpe2015sign} and the MIV-SPS ellipsoidal outer approximation computation, apart from the reduced assumptions on the noise vectors, is the size of the constraint matrix in the resulting convex programming formulations. Let $L \in \mathbb{R}^{l\times l}$ be the constraint matrix we get in the direct LSS identification problem. Then, in the IV-SPS case, $l$ grows according to $l=d_x^2 + d_xd_u + 1$ with the dimensions of the inputs and the outputs. On the other hand, in the MIV-SPS case, it just grows according to $l=2\,d_x + d_u$ with the dimensions.

\vspace{-1mm}
\section{Numerical Experiments}\label{sec:experiments}
In this section we present numerical experiments evaluating MIV-SPS for a closed-loop LSS system. The ellipsoidal outer approximations of MIV-SPS are compared to that of IV-SPS, after suitable vectorization.
The effect of using different exploitation-exploration trade-offs in the feedback is also studied experimentally. Finally, the sample efficiency of the ellipsoid over-bound confidence regions of MIV-SPS are investigates for different feedback rules.
\subsection{Experimental Setup}\label{sec:exp_setup}
We consider the following closed-loop LSS system
\vspace{-1mm}
\begin{equation}
\label{exp_setting}
\begin{aligned}
	x_{k+1}\,  =\;\,& Ax_k \,+\, Bu_k \,+\, w_k,\\[1mm]
	u_k\,  =\; \,& \varepsilon Kx_k \,+\, (1-\varepsilon)r_k,
\end{aligned}
\vspace{-1mm}
\end{equation}
 {\newtext where $A$ is a stable 
matrix (i.e., its eigenvalues are inside the unit circle), the elements of the square matrix $B$} are sampled from the uniform distribution $\mathcal{U}(1, 10)$ and $K$ is the optimal controller for $\varepsilon = 1$, given the quadratic cost
\vspace{-1mm}
{\newtext 
\begin{equation}
	J(K) = \sum_{k=0}^{n-1} q\norm{x_k}^2 + v \norm{u_k}^2,
	\vspace{-0.5mm}
\end{equation}
i.e., an LQR problem \cite{Bertsekas2007}.
In the cost function the values of the scalar cost weights are $q = v = 1$, unless stated otherwise.}
The scalar $\varepsilon \in (0,1)$ is the {\em exploitation} parameter. Sequences $\{w_k\}$ and $\{r_k\}$ contain independent random vectors, where the references are standard normal, 
while the noise distributions vary with the experiments.

We consider a finite data sample of size $n$ that contains input-output-reference tuples, $\{\langle u_k, x_k, r_k\rangle\}_{k = 0}^{n-1}$.

The instrumental variables $\{\psi_k\}$
are generated from the data.
We present the direct and the indirect cases side-by-side.
First, least squares estimates of the parameter matrices
are computed.
The system parameters $A$ and $B$ are estimated in the direct case, while in the indirect case $C$ and $D$ are estimated. Let $\widehat{\Theta}^{\scriptscriptstyle \text{LS}}$ and  $\widehat{\Theta}^{\scriptscriptstyle \text{LS}}_{\scriptscriptstyle \text{id}}$ denote the LS estimate in the direct and the indirect case, respectively:
\begin{align}
	&\widehat{\Theta}^{\scriptscriptstyle \text{LS}}\, \defeq\, (\Phi\tr\Phi)^{-1}\Phi\tr Y,\quad
	&\widehat{\Theta}^{\scriptscriptstyle \text{LS}}_{\scriptscriptstyle \text{id}}\, \defeq\, (\Phi_{\scriptscriptstyle \text{id}}\tr\Phi_{\scriptscriptstyle \text{id}})^{-1}\Phi_{\scriptscriptstyle \text{id}}\tr Y,
\vspace{-2mm}	
\end{align}
where
\vspace{-1mm}
\begin{align}
	&\widehat{\Theta}^{\scriptscriptstyle \text{LS}} \, \defeq\, \begin{bmatrix}
		\hat{A}\tr \\
		\hat{B}\tr
	\end{bmatrix}\!,
	&\widehat{\Theta}^{\scriptscriptstyle \text{LS}}_{\scriptscriptstyle \text{id}} \, \defeq\, \begin{bmatrix}
		\hat{C}\tr \\
		\hat{D}\tr
	\end{bmatrix}\!.
\end{align}
Using the estimates $\hat{A}$ and $\hat{B}$, or $\hat{C}$ and $\hat{D}$ in the indirect case, noiseless state sequences are generated as follows
\begin{align}
	&\bar{x}_{k+1} \, \defeq\, \hat{A}\bar{x}_k + \hat{B}r_k,\quad
	&\bar{x}_{k+1}^{\scriptscriptstyle \text{id}} \, \defeq\, \hat{C}\bar{x}_k + \hat{D}r_k.
\end{align}
Using the noiseless state sequences, the instrumental variables can be defined in the direct and indirect cases as
\begin{align}\label{ivinexperiment}
	&\psi_k\, \defeq\, (\bar{x}_k\tr, r_k\tr)\tr\!\!,\quad &\psi_k^{\scriptscriptstyle \text{id}}\, \defeq\, ((\bar{x}_k^{\scriptscriptstyle \text{id}})\tr\!\!, r_k\tr)\tr\!\!.
\end{align}

In the strict sense, the 
IVs defined in \eqref{ivinexperiment} are not completely independent of the noise, because the least squares estimates $\widehat{\Theta}^{\scriptscriptstyle \text{LS}}$ and $\widehat{\Theta}^{\scriptscriptstyle \text{LS}}_{\scriptscriptstyle \text{id}}$ depend on the noise. However, in both estimates the noises are ``averaged out'', hence their effect is minimal. If the 
LS estimates were constructed from a data sample that is independent of the 
sample used by the SPS algorithm, then the obtained regions would be {\newtext guaranteed to have exact coverage.}
The difference between {\newtext these} two constructions considering the experimental results is negligible, therefore we only used one data sample.

In the experiments the conservatism of the ellipsoidal outer approximations of the confidence regions {\newtext is} investigated using Monte Carlo simulations. In these simulations we set the confidence probability to $p = 0.9$ and computed an empirical probability $\hat{p} = \nicefrac{s_{\scriptscriptstyle \text{in}}}{s}$, where $s_{\scriptscriptstyle \text{in}}$ is the number of simulations, where the true parameter matrix $\Theta^*$ is included in the ellipsoid, and $s$ is number of simulations.
\subsection{Comparison of MIV-SPS and IV-SPS}
As discussed in Section \ref{sec:cvxprog}, a crucial difference between the MIV-SPS and IV-SPS ellipsoidal outer approximation (EOA) computations is the size of the constraint matrices in the 
SDP formulations. As a consequence, solving the convex program in the MIV-SPS case should be faster than in the IV-SPS case. 
Table \ref{table:comparsion_ls_mv_efficiency} presents the number of flops and the required CPU time for solving the convex programs in the MIV-SPS EOA {\em relative} to those of the IV-SPS EOA, for different dimensions.
{\newtext The obtained results confirm that EOAs can be computed more efficiently by the MIV-SPS approach, especially in higher dimensions.}
\begin{table}[t]
	\caption{Flops and CPU time needed to compute the ellipsoidal outer approximation of MIV-SPS relative to IV-SPS, $n=500$, $s = 500$.}
	\vspace{2mm}
	\label{table:comparsion_ls_mv_efficiency}
	\setlength{\tabcolsep}{3pt}
	\centering
	\begin{tabular}{|p{45pt}|p{45pt}|p{55pt}|p{55pt}|}
		\hline
		Dim.&
		Params.&
		Rel. flops&
		Rel. time \\
		\hline
		\hline
		1&
		2& 
		0.839& 
		1.03\\
		\hline
		2&
		8& 
		0.427& 
		0.923\\
		\hline
		3&
		18& 
		0.146& 
		0.528\\
		\hline
		4&
		32&
		0.056& 
		0.195\\
		\hline
	\end{tabular}
	\vspace{-2mm}
\end{table}

The experiments discussed in this section, i.e., Tables \ref{table:comparsion_ls_mv_efficiency}, \ref{table:comparsion_ls_mv_p_normnoise}, \ref{table:comparsion_ls_mv_p_bimodnormnoise} and \ref{table:comparsion_ls_mv_p_nonstatnoise}, are all based on $s = 500$ Monte Carlo simulations, for each dimension, with sample size $n=500$ using an open-loop system ($\varepsilon=0$) and direct identification.

Although, strong duality is not proven yet for the reformulation of \eqref{cvxprog}, only weak duality is exploited, the experimental evidences are indicative of the phenomenon that it holds. The empirical coverage probabilities (i.e., the ratio of cases when the regions cover the true parameter) in the MIV-SPS EOA case $\hat{p}_{\scriptscriptstyle \text{MIV}}$ are very close to those of IV-SPS EOA $\hat{p}_{\scriptscriptstyle \text{IV}}$, for which strong duality is known. 

We have compared the empirical coverage probabilities $\hat{p}$ of the asymptotic 
ellipsoids $\hat{p}_{\scriptscriptstyle \text{AS}}$, 
see Section \ref{sec:lsandasym}, the empirical probabilities $\hat{p}_{\scriptscriptstyle \text{IN}}$ for the MIV-SPS Indicator algorithm, presented in Section \ref{sec:spsindicator} {\newtext (the ratio of cases when the Indicator algorithms 
return 1 for the true parameter)}, and the empirical probabilities $\hat{p}_{\scriptscriptstyle \text{IV}}$, $\hat{p}_{\scriptscriptstyle \text{MIV}}$ for three 
noise scenarios. The target confidence probability was always $0.9$. 

The empirical coverage probabilities in case the process noise $\{w_k\}$ is a sequence of i.i.d.\ (multivariate) standard normal vectors is given in Table \ref{table:comparsion_ls_mv_p_normnoise}. It can be observed that the asymptotic and the MIV-SPS Indicator confidence regions are very precise and the SPS confidence ellipsoidal outer approximations become more conservative as the dimension increases, but the empirical coverage probabilities of MIV-SPS and IV-SPS stay close to each other.
\begin{table}[t]
	\caption{Comparison of empirical probabilities of covering the true parameter matrix for standard normal noises, $n=500$, $s = 500$.}
	\vspace{2mm}
	\label{table:comparsion_ls_mv_p_normnoise}
	\setlength{\tabcolsep}{3pt}
	\centering
	\begin{tabular}{|p{36pt}|p{36pt}|p{36pt}|p{36pt}|p{36pt}|p{36pt}|}
		\hline
		Dim.&
		Params.&
		$\hat{p}_{\scriptscriptstyle \text{AS}}$&
		$\hat{p}_{\scriptscriptstyle \text{IN}}$&
		$\hat{p}_{\scriptscriptstyle \text{IV}}$&
		$\hat{p}_{\scriptscriptstyle \text{MIV}}$  \\
		\hline
		\hline
		1&
		2& 
		0.892&
		0.894&
		0.926&
		0.926\\
		\hline
		2&
		8& 
		0.892&
		0.894&
		0.954&
		0.946\\
		\hline
		3&
		18&
		0.89&
		0.888&
		0.982&
		0.978\\
		\hline
		4&
		32&
		0.902&
		0.89&
		0.996&
		0.996\\
		\hline
	\end{tabular}
    \vspace{-2mm} 
\end{table}
\begin{table}[t]
	\caption{Comparison of empirical probabilities of covering the true parameter for 
	mixture of two Gaussian noises, $n=500$, $s = 500$.}
	\vspace{2mm}
	\label{table:comparsion_ls_mv_p_bimodnormnoise}
	\setlength{\tabcolsep}{3pt}
	\centering
	\begin{tabular}{|p{36pt}|p{36pt}|p{36pt}|p{36pt}|p{36pt}|p{36pt}|}
		\hline
		Dim.&
		Params.&
		$\hat{p}_{\scriptscriptstyle \text{AS}}$&
		$\hat{p}_{\scriptscriptstyle \text{IN}}$&
		$\hat{p}_{\scriptscriptstyle \text{IV}}$&
		$\hat{p}_{\scriptscriptstyle \text{MIV}}$  \\
		\hline
		\hline
		1&
		2& 
		0.9&
		0.9&
		0.932&
		0.932\\
		\hline
		2&
		8& 
		0.876&
		0.9&
		0.95&
		0.958\\
		\hline
		3&
		18&
		0.852&
		0.908&
		0.984&
		0.98\\
		\hline
		4&
		32&
		0.847&
		0.898&
		0.982&
		0.982\\
		\hline
	\end{tabular}
    \vspace{-2mm} 
\end{table}
\begin{table}[t]
	\caption{Comparison of empirical coverage probabilities for bimodal mixture of two non-stationary Laplacian noises, $n=500$, $s = 500$.}
	\vspace{2mm}
	\label{table:comparsion_ls_mv_p_nonstatnoise}
	\setlength{\tabcolsep}{3pt}
	\centering
	\begin{tabular}{|p{36pt}|p{36pt}|p{36pt}|p{36pt}|p{36pt}|p{36pt}|}
		\hline
		Dim.&
		Params.&
		$\hat{p}_{\scriptscriptstyle \text{AS}}$&
		$\hat{p}_{\scriptscriptstyle \text{IN}}$&
		$\hat{p}_{\scriptscriptstyle \text{IV}}$&
		$\hat{p}_{\scriptscriptstyle \text{MIV}}$  \\
		\hline
		\hline
		1&
		2& 
		0.888&
		0.896&
		0.928&
		0.928\\
		\hline
		2&
		8& 
		0.872&
		0.906&
		0.942&
		0.96\\
		\hline
		3&
		18&
		0.846&
		0.908&
		0.968&
		0.976\\
		\hline
		4&
		32&
		0.834&
		0.898&
		0.976&
		0.982\\
		\hline
	\end{tabular}
    \vspace{-2mm}
\end{table}

In the case, where the process noise $\{w_k\}$ is a sequence of i.i.d.\ Gaussian random vectors, the asymptotic confidence region gives accurate estimation. SPS builds confidence regions under milder conditions, even when the noise is non-stationary. Simulation experiments were {\newtext performed}, where $\{w_k\}$ is a sequence of i.i.d. bimodal mixture of two Gaussians $w_k = \frac{1}{2}Z_1 + \frac{1}{2}Z_2, Z_1 \sim \mathcal{N}(1, \sigma_wI), Z_2 \sim \mathcal{N}(-1, \sigma_wI)$, and where $\{w_k\}$ is a sequence of time dependent mixture of two Laplacians $w_k = \frac{1}{2}Z_3 + \frac{1}{2}Z_4, Z_3 \sim \mathcal{L}(\frac{5(k+1)}{n}, (\frac{(k+1)}{n} + \sigma_w)I), Z_4 \sim \mathcal{L}(\frac{-5(k+1)}{n}, (\frac{(k+1)}{n} + \sigma_w)I)$, with $\sigma_w = 1$.

The results of these experiments are presented in Tables \ref{table:comparsion_ls_mv_p_bimodnormnoise} and \ref{table:comparsion_ls_mv_p_nonstatnoise}. In both experiments the empirical coverage probabilities of the asymptotic confidence regions are less than $p = 0.9$ by a significant margin. In the bimodal Gaussian case, this deviation is due to the insufficient sample size, while in the bimodal Laplacian noise case, it is also due to the heavier-tail and the non-stationarity of the noise. 

On the other hand, the MIV-SPS Indicator algorithm constructs proper confidence regions in both cases, and the ellipsoidal outer approximations give similar coverage probabilities in the IV-SPS and MIV-SPS cases. The ellipsoidal outer approximations introduce some conservatism, but compared to the confidence regions based on the asymptotic theory, they provide better results when a non-asymptotic
lower bound on the probability is required.
\subsection{Exploration vs Exploitation}
\label{sec:exp_exploit}
The signal $\{r_k\}$ can be seen as an exploration noise affecting the controller instead of a reference signal. In this case the parameter $\varepsilon$ controls the exploitation-exploration trade-off. The system operates in an open-loop setting when $\varepsilon = 0$ and without exploration in a closed-loop when $\varepsilon = 1$. In our experimental setting, given in \eqref{exp_setting}, the choice of exploitation rate affects the size of the EOA of the confidence region in a way that with more exploitation the outer approximation of the confidence region gets more conservative. This phenomenon is due to the construction of our instrumental variables detailed in Section \ref{sec:exp_setup}. The instrumental variables are created by generating a noiseless state sequence using the exploration noise $\{r_k\}$. With less exploration, the correlation between the instrumental variable sequence and the state sequence is decreasing, thus the IV estimate gets less accurate, therefore the EOA gets more conservative. In a real-world settings, instrumental variables can often be created from signals that are observable, but not controllable. In such cases no conservatism would be introduced with more exploitation. 

The choice of the controller $K$ also influences the conservatism of the EOA. We have investigated the empirical coverage probabilities of MIV-SPS EOAs {\newtext for different $q$ and $v$ instances} using the direct and indirect identification approaches. The results are illustrated in Figure \ref{fig:exploitation_dir} for the direct case and in Figure \ref{fig:exploitation_indir} for the indirect case. The results are based on $s = 500$ Monte Carlo simulations for each $\varepsilon$ on a 2-dimensional LSS system with sample size $n=500$. Standard normal process noise is applied.
\begin{figure}[!t]
	\centerline{\includegraphics[width=\columnwidth]{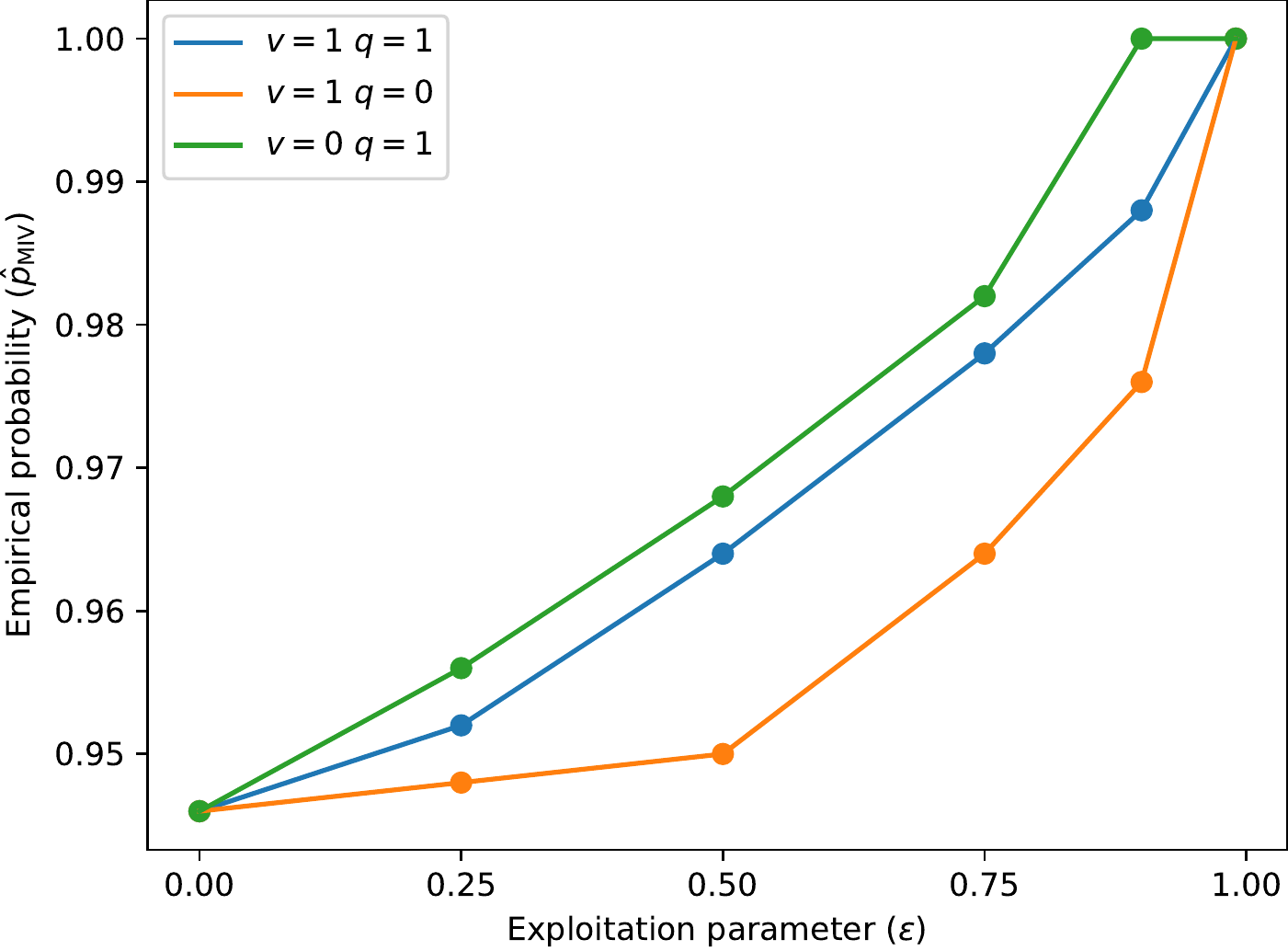}}
	\caption{Empirical coverage probability of MIV-SPS EOA as a function of the exploitation parameter for different feedback setups, using {\em direct} identification, for a 2-dimensional LSS, $n=500$, $s=500$.}
	\label{fig:exploitation_dir}
    \vspace{2mm}
\end{figure}
\begin{figure}[!t]
	\centerline{\includegraphics[width=\columnwidth]{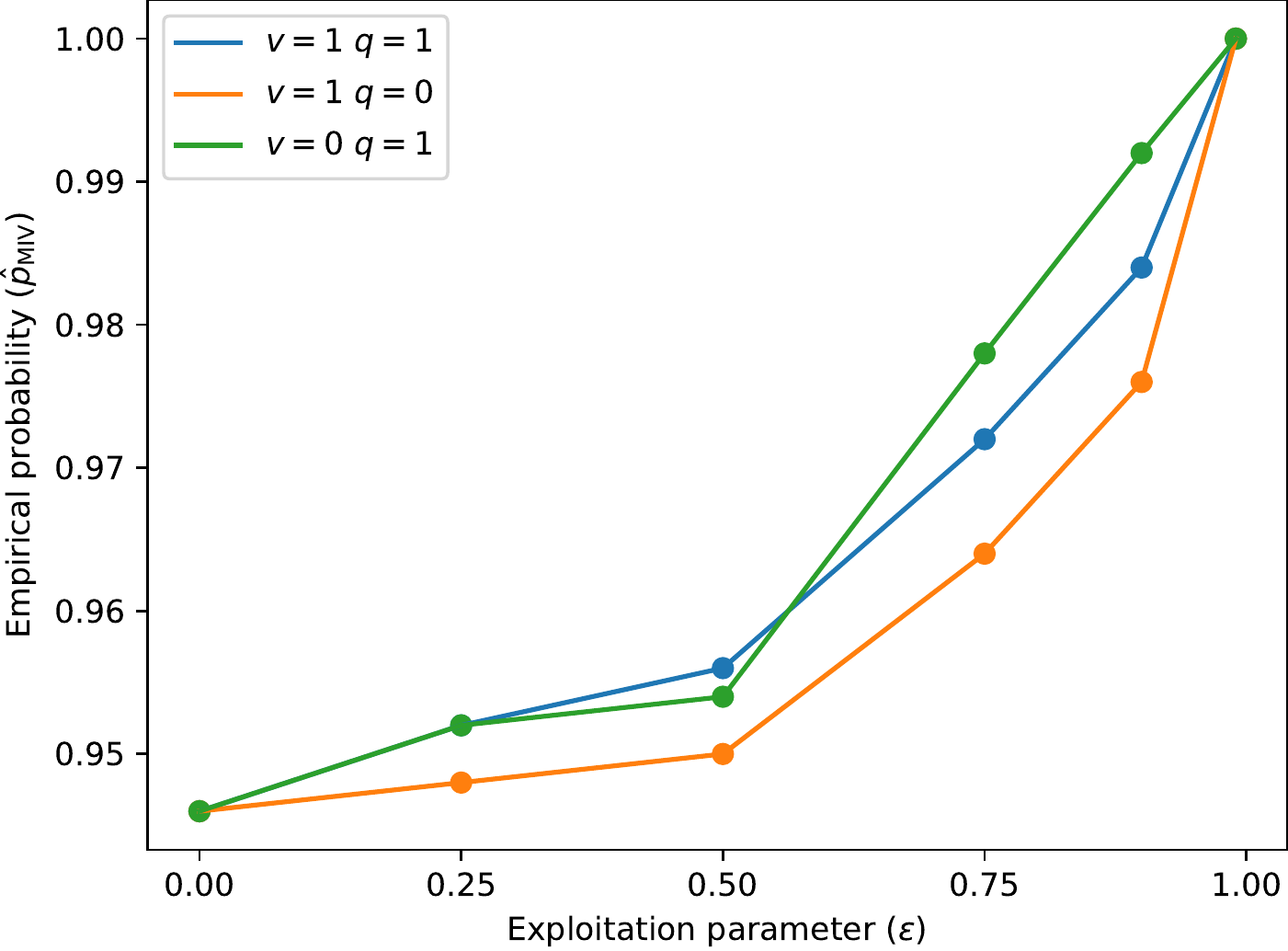}}
	\caption{Empirical coverage probability of MIV-SPS EOA as a function of the exploitation parameter for different feedback setups, using {\em indirect} identification, for a 2-dimensional LSS, $n=500$, $s=500$.}
	\label{fig:exploitation_indir}
\end{figure}
\subsection{Sample Complexity}
In the last experiments, the sample complexity of the MIV-SPS EOA method is investigated. As the sample size $n$ increases, the 
EOA becomes less conservative, hence it provides more accurate confidence regions. Simulation results are presented for different exploitation parameter $\varepsilon$ values. $500$ Monte Carlo simulations were run for each sample size 
on a 2-dimensional LSS system, using standard normal noise and with direct identification. The results are illustrated in Figure \ref{fig:samplecomplex}. It can be observed that with more exploitation (as $\varepsilon$ increases), the empirical coverage probability of EOA converges slower to the target confidence, due to the inaccuracy of the IV estimate.

\begin{figure}[!t]
	\centerline{\includegraphics[width=\columnwidth]{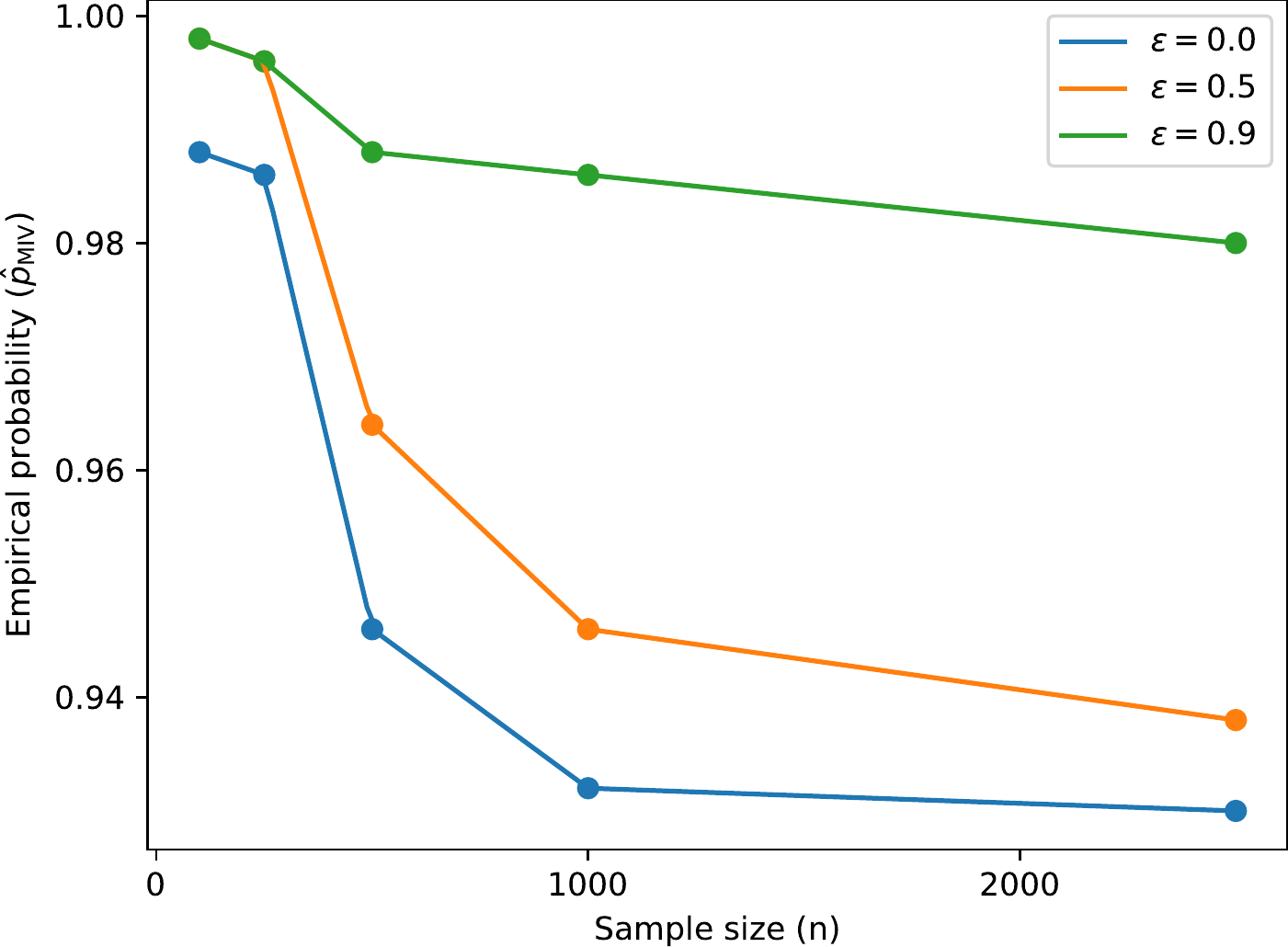}}
	\caption{Empirical coverage probability of MIV-SPS EOA as a function of the sample size $n$ in an open-loop ($\varepsilon = 0$) setting, using direct identification, for a 2-dimensional LSS, $s=500$.}
	\label{fig:samplecomplex}
\end{figure}

\vspace{-1mm}
\section{Conclusions}\label{conclusion}
The paper introduced a {\em matrix variate regression} and {\em instrumental variables} based generalization of the {\em Sign-Perturbed Sums} (SPS) identification method, called MIV-SPS.
It was illustrated on {\em closed-loop linear state-space} (LSS) models.
The MIV-SPS method constructs {\em non-asymptotic}, {\em distribution-free} confidence regions around the instrumental variable estimate under mild 
statistical 
assumptions on the process noise. The {\em exact confidence} and the {\em strong consistency} of 
MIV-SPS regions were proven.

A computationally efficient {\em ellipsoidal outer approximation} (EOA) method was suggested, as well, which is a compact over-bound of the MIV-SPS confidence region. Simulation experiments demonstrated that, for LSS models, our proposed method computes the EOA more efficiently than the scalar IV-SPS, developed earlier, while keeping a similar size. Experiments comparing MIV-SPS and the confidence ellipsoids based on the asymptotic theory were also concluded, showing the superiority of MIV-SPS, e.g., for {\em heavy-tailed} and {\em non-stationary} noise sequences.

Extending MIV-SPS to be able to handle {\em partially observable} systems and combining the ideas with {\em robust control} approaches are subject of future research.

\vspace{-1mm}
\section*{Acknowledgments}\label{acknowledgments}
This research was supported by the European Union within the framework of the National Laboratory for Autonomous Systems, RRF-2.3.1-21-2022-00002; and by the TKP2021-NKTA-01 grant of the National Research, Development and Innovation Office (NRDIO), Hungary.

\vspace{-1mm}
\bibliography{iv-mimo-clss-sps}
\appendix

\section{Semidefinite Programming Formulation of the Outer Approximation}\label{app:cvxformulation}
The optimization problem
\begin{equation}
	\begin{aligned}
		\max \quad & \norm{Z}^2\frob\\
		\textrm{s.t.} \quad &\Tr(Z\tr A_iZ + Z\tr B_i + B_i\tr Z + C_i) \leq 0, \\
	\end{aligned}
\end{equation}
has the equivalent form
\begin{equation}\label{equ:optprob2}
	\begin{aligned}
		\min \quad & -\Tr(Z\tr Z)\\
		\textrm{s.t.} \quad &\Tr(Z\tr A_iZ + Z\tr B_i + B_i\tr Z + C_i) \leq 0. \\
	\end{aligned}
\end{equation}
The Lagrangian of the problem \eqref{equ:optprob2} can be written as
\begin{equation}
	\begin{aligned}
		&L(Z, \lambda)\, \defeq\, \\&-\Tr(Z\tr Z) + \lambda\Tr(Z\tr A_iZ + Z\tr B_i + B_i\tr Z + C_i)\\& = \Tr(Z\tr(-I+\lambda A_i)Z +\lambda Z\tr B_i + \lambda B_i\tr Z + \lambda C_i).
	\end{aligned}
\end{equation}
The Lagrange dual function can be formulated as
\begin{equation}
	\begin{aligned}
		g(\lambda) \,\defeq \, \inf_Z L(Z, \lambda).
	\end{aligned}
\end{equation}
If $L(Z, \lambda)$ takes its minimum with respect to $Z$, then
\begin{equation}
	\begin{aligned}
		\nabla_Z L(Z, \lambda) \,=\, (-I+\lambda A_i)Z + \lambda B_i = 0.
	\end{aligned}
\end{equation}
From the above equation it can be seen that $L(Z, \lambda)$ takes its minimum value, where $Z = -(-I+\lambda A_i)^{-1}\lambda B_i$, thus the Lagrange dual function can be written as
\begin{equation}
	\begin{aligned}
		g(\lambda) \,=\, &\Tr(-\lambda^2B_i\tr (-I+\lambda A_i)^{-1}B_i + \lambda C_i),
	\end{aligned}
\end{equation}
therefore, the dual problem can be formulated as
\begin{equation}\label{equ:dualprob}
	\begin{aligned}
		\max \quad & \Tr(-\lambda^2B_i\tr (-I+\lambda A_i)^{-1}B_i + \lambda C_i)\\
		\textrm{s.t.} \quad&\lambda \geq 0.
	\end{aligned}
\end{equation}
The optimization problem \eqref{equ:dualprob} can be reformulated as
\begin{equation}
	\begin{aligned}
		\min \quad & \Tr(\Gamma - \lambda C_i)\\
		\textrm{s.t.} \quad&\lambda \geq 0 \\
		\quad&\Gamma \succeq \lambda^2B_i\tr (-I+\lambda A_i)^{-1}B_i,
	\end{aligned}
\end{equation}
where $\Gamma$ is a symmetric matrix and relation ``$\succeq$'' 
denotes the Loewner partial ordering of semidefinite matrices. 

Finally, the constraint $\Gamma \succeq \lambda^2B_i\tr (-I+\lambda A_i)^{-1}B_i$ can be reformulated by resorting to the Schur complement,
\begin{equation}
	\begin{bmatrix}
		-I + \lambda A_i & \lambda B_i \\
		\lambda B_i\tr & \Gamma\\
	\end{bmatrix} \succeq 0,
\end{equation}
which leads to the convex optimization problem \eqref{cvxprog}.
\section{Proof of Theorem \ref{theorem1}\\ (Exact Confidence)}\label{app:proof_theorem1}
\begin{definition}\label{definition1}
	Let $Z_1,\dots,Z_k$ be a finite collection of
	random variables and $\succ$ a strict total order. If for all
	permutations $i_1,\dots,i_k$ of indices $1,\dots,k$ we have
	\vspace{-1mm}
	\begin{equation}
		\mathbb{P}(Z_{i_k} \succ Z_{i_{k-1}} \succ \dots \succ Z_{i_1}) = \frac{1}{k!},
    \vspace{-1mm}		
	\end{equation}
	then we call $\{Z_i\}$ uniformly ordered w.r.t. order $\succ$.
\end{definition}
\begin{lemma}\label{lemma1}
	Let $\alpha,\beta_1,\dots,\beta_k$ be i.i.d. random signs (random variables taking $\pm 1$ with probability $\nicefrac{1}{2}$ each), then the variables $\alpha,\alpha \cdot \beta_1,\dots,\alpha \cdot \beta_k$ are i.i.d. random
	signs.
\end{lemma}
\begin{lemma}\label{lemma2}
	Let $X$ and $Y$ be two independent, $\mathbb{R}^d$-valued and $\mathbb{R}^k$-valued random vectors, respectively. Let us consider a (measurable) function $g : \mathbb{R}^d \times \mathbb{R}^k \rightarrow \mathbb{R}$ and a (measurable) set $A \subseteq \mathbb{R}$. If we have $\mathbb{P}(g(x, Y) \in A) = p$, for all (constant) $x \in \mathbb{R}^d$, then we also have $\mathbb{P}(g(X, Y) \in A) = p$.
\end{lemma}
\begin{lemma}\label{lemma3}
	Let $Z_1,\dots,Z_k$ be real-valued, i.i.d. random variables. Then, they are uniformly ordered w.r.t. $\succ_{\pi}$.
\end{lemma}
\noindent The proofs of Lemmas \ref{lemma1}, \ref{lemma2} and \ref{lemma3} can be found in \cite{Csaji2015}.
\begin{lemma}\label{lemma4}
    Let $w$ be a random vector with $w\; {\buildrel {\scriptscriptstyle d} \over =}\,-w$, that is $w$ is symmetrically distributed about zero. Then, there exist a random sign $s$ and a random vector $v$ with the properties: $s$ is independent of $v$ and $w = s\cdot v$.
\end{lemma}
\begin{proof}[Proof of Lemma \ref{lemma4}]
    Let $s$ be a random sign (that is: a Rademacher random variable) which is independent of $w$, and let $v \defeq s\cdot w$. Then, it holds for every event $A$, that
    \vspace{-1mm}
    \begin{equation}
	\begin{aligned}
    \mathbb{P}(s = 1 \wedge v \in A) & \,=\, \mathbb{P}(s = 1 \wedge s\cdot w \in A)\\
                                     & \,=\,\mathbb{P}(s = 1 \wedge w \in A)\\
                                     & \,=\, \mathbb{P}(s = 1)\cdot\mathbb{P}(w \in A)\\
                                     & \,=\, \mathbb{P}(s = 1)\cdot\mathbb{P}(s\cdot w \in A)\\
                                     & \,=\, \mathbb{P}(s = 1)\cdot\mathbb{P}(v \in A).
    \end{aligned}
    \end{equation} 
    \begin{equation}
	\begin{aligned}
    \mathbb{P}(s = -1 \wedge v \in A) & \,=\, \mathbb{P}(s = -1 \wedge s\cdot w \in A)\\
                                     & \,=\,\mathbb{P}(s = -1 \wedge -w \in A)\\
                                     & \,=\, \mathbb{P}(s = 1)\cdot\mathbb{P}(-w \in A)\\
                                     & \,=\, \mathbb{P}(s = 1)\cdot\mathbb{P}(s\cdot w \in A)\\
                                     & \,=\, \mathbb{P}(s = 1)\cdot\mathbb{P}(v \in A). 	    
    \vspace{-1mm}                                     
	\end{aligned}
    \end{equation} 	
	Therefore, $s$ and $v$ are independent and $w = s\cdot v$.
\end{proof}
\begin{proof}[Proof of Theorem \ref{theorem1}]
	By definition, the true parameter, $\Theta^*$, is included in the confidence region if $\mathcal{R}(\Theta^*) \leq m-q$. In this case $\norm{S_0(\Theta^*)}\frob^2$ is ranked $1,\dots,m-q$ according to the ascending order (w.r.t. $\succ_{\pi}$) of variables $\{ \norm{S_i(\Theta^*)}\frob^2\}$. 
	
	We will prove that $\{\norm{S_i(\Theta^*)}\frob^2\}$ are uniformly ordered, thus the probability that $\norm{S_0(\Theta^*)}\frob^2$ takes any position in the ordering equals $1/m$, therefore the probability that its rank is at most $m-q$ is $1-q/m$. 
	Throughout the proof we will work with a {\em fixed realization} of the instrumental variables, the general case follows from  Lemma \ref{lemma2}.

	For $\Theta = \Theta^*$, the $S_i(\cdot)$ functions can be formulated as
	\vspace{-1mm}
	\begin{equation}\label{equ:Si_proof}
		S_i(\Theta^*) \,=\, \tfrac{1}{n}P_n^{-\frac{1}{2}}\Psi\tr\Lambda_iW,
	\vspace{-1mm}		
	\end{equation}
	$\forall\, i \in \{0,\dots,m-1\}$, where $\Lambda_0 = I \in \R^{n \times n}$.
	
	From \eqref{equ:Si_proof} it can be seen that functions $\{S_i(\cdot)\}$ depend on the perturbed noise matrices, $\{\Lambda_iW\}$, via the same function for all $i$, which is denoted by $S(\Lambda_iW) \defeq S_i(\Theta^*)$.
	
	By using the result of Lemma \ref{lemma4}, the noise vector sequence $\{w_k\}$ can be decomposed as $w_k = s_k\tilde{w}_k$, where $s_k$ is a Rademacher random variable independent of $\tilde{w}_k$. Let $\gamma_{i,k} \defeq \alpha_{i,k}s_k$, where $\alpha_{0,k} = 1$, $\forall k,\forall i \in {0,\dots,m-1}$. Lemma \ref{lemma1} can be used to show that $\gamma_{i,k}$ are i.i.d. random signs, since $\{\alpha_{i,k}\}$ are random signs and $\tilde{s}_k$ is a Rademacher random variable. As a consequence of Lemma \ref{lemma4}, the sequence $\{\gamma_{i,k}\}$ are independent of $\{\tilde{w}_k\}$.
	
	By fixing a realization of $\{\tilde{w}_k\}$, called $\{v_k\}$, real-valued variables $\{Z_i\}$ can be defined as 
	\vspace{-1mm}
	\begin{equation}
		Z_i \,\defeq\, \norm{S(\Gamma_iM)}\frob^2\!,
    \vspace{-1mm}		
	\end{equation}
	where
	\vspace{-1mm}
	\begin{equation}
		\Gamma_i \defeq \begin{bmatrix}
			\gamma_{i,1} & & \\
			& \ddots & \\
			& & \gamma_{i,n}
		\end{bmatrix}\!,\qquad
		M \defeq \begin{bmatrix}
			v_0\tr \\
			v_1\tr \\
			\vdots \\
			v_{n-1}\tr
		\end{bmatrix}\!.
    \vspace{-1mm}  
	\end{equation}
	
	Random variables $\{Z_i\}$ are i.i.d., and by Lemma \ref{lemma3}, it follows that they are uniformly ordered w.r.t.\ $\succ_{\pi}$.
	
	In order to relax the assumption of a fixed realization of $\{\tilde{w}_k\}$, Lemma \ref{lemma2} can be applied. By doing so, it follows that the $\{\norm{S_i(\Theta^*)}\frob^2\}$ variables are uniformly ordered (unconditionally), which completes the proof of the theorem.
\end{proof}
\vspace{-2mm}
\section{Proof of Theorem \ref{theorem2}\\
(Strong Consistency)}\label{app:proof_theorem2}
\begin{proof}[Proof of Theorem \ref{theorem2}]
    The proof consists of two parts. In the first part we are going to prove that for any
    $\Theta' \neq \Theta^*$: 
    \vspace{-1mm}    
    \begin{equation}
    \norm{S_0(\Theta')}\frob^2\, \xrightarrow{a.s.}\, \norm{P^{-\frac{1}{2}}V\tilde{\Theta}}\frob^2 \,>\, 0,
    \vspace{-1mm}
    \end{equation}
    where $\tilde{\Theta} \defeq (\Theta^* - \Theta')$ and matrices $P$ and $V$ are defined in Assumptions \ref{assu:psd_p} and \ref{assu:inv_v}. On the other hand, for $i \neq 0$, we have
    \vspace{-1mm}
    \begin{equation}
    \norm{S_i(\Theta')}\frob^2\, \xrightarrow{a.s.} 0,\,
    \vspace{-1mm}
    \end{equation}
as {\newtext $n \to \infty$}. Consequently, as $n$ grows, the rank $\mathcal{R}(\Theta')$ of $\norm{S_0(\Theta')}\frob^2$ will be eventually equal to $m$, therefore $\Theta'$ will be (a.s.) excluded from the confidence region, as {\newtext $n \to \infty$}. 

The results will be derived for a fixed realization of the instrumental variables, just like in the proof of Theorem \ref{theorem1}. Since the matrix of IVs $\Psi_n$ and the noise $W$ are independent \eqref{assu:iv1}, the obtained results hold true (almost surely) on the whole probability space, ensured by
    Lemma \ref{lemma2}. In the second part, we will prove that the confidence region converge to $\Theta^*$ uniformly, not just pointwise.

 $S_0(\Theta')$ can be formulated as
     \vspace{-1mm}
	\begin{align}
			S_0(\Theta') & = \tfrac{1}{n}P_n^{-\frac{1}{2}}\Psi_n\tr \CE_n(\Theta')
			= \tfrac{1}{n}P_n^{-\frac{1}{2}}\Psi_n\tr(Y_n-\Phi_n\Theta') \notag\\
			&= \tfrac{1}{n}P_n^{-\frac{1}{2}}\Psi_n\tr\Phi_n\tilde{\Theta} + \tfrac{1}{n}P_n^{-\frac{1}{2}}\Psi_n\tr W_n.
       \vspace{-1mm}
	\end{align}
	The two terms can be examined separately. By observing that $(\cdot)^{\frac{1}{2}}$ is a continuous matrix function and applying \ref{assu:psd_p} and \ref{assu:inv_v} the convergence of the first part follows, thus
    \vspace{-1mm}
	\begin{equation}
		\tfrac{1}{n}P_n^{-\frac{1}{2}}\Psi_n\tr\Phi_n\tilde{\Theta} = P_n^{-\frac{1}{2}}V_n\tilde{\Theta}\, \xrightarrow{a.s.}\, P^{-\frac{1}{2}}V\tilde{\Theta},
        \vspace{-1mm}
	\end{equation}
	as {\newtext $n \to \infty$}. The convergence of the second term is proved from the element-wise application of Kolmogorov’s strong
	law of large numbers (SLLN) for independent variables \cite{Shiryaev1997}. Note that $\{P_n^{-\frac{1}{2}}\} \xrightarrow{a.s.} \{P^{-\frac{1}{2}}\}$ as $n \to \infty$, thus it is enough to prove that $\tfrac{1}{n}\Psi_n\tr W_n \xrightarrow{a.s.} 0$. By applying the Cauchy-Schwarz inequality, \ref{assu:i_grr}, and \ref{assu:var_grr}, we have
	\vspace{-1mm}
	\begin{equation}\label{equ:CS-Kolmogorov}
		\begin{aligned}
			&\sum_{k=0}^{\infty}\frac{\mathbb{E}\!\left[(\psi_{k,j}w_{k,l})^2\right]}{k^2} \leq
			\sum_{k=0}^{\infty}\frac{\norm{\psi_{k}}^2}{k}\frac{\mathbb{E}\!\left[\norm{w_{k}}^2\right]}{k} \leq \\
			&\sqrt{\sum_{k=0}^{\infty}\frac{\norm{\psi_{k}}^4}{k^2}}\sqrt{\sum_{k=0}^{\infty}\frac{\mathbb{E}\!\left[\norm{w_{k}}^2\right]^2}{k^2}} < \infty.
		\end{aligned}
	\end{equation}
	Therefore the condition of the SLLN holds, thus
	\vspace{-1mm}
	\begin{equation}
		\tfrac{1}{n}P_n^{-\frac{1}{2}}\Psi_n\tr W_n \xrightarrow{a.s.} 0,\quad \text{ as }\quad n \to \infty.
	\end{equation}
	From the above results we obtain
	\vspace{-1mm}
	\begin{equation}
		\begin{aligned}
			&\norm{S_0(\Theta')}\frob^2 \xrightarrow{a.s.}
			\norm{P^{-\frac{1}{2}}V\tilde{\Theta}}\frob^2 > 0,
		\end{aligned}
	\end{equation}
	since $V$ is full rank, $P^{-\frac{1}{2}}$ is positive definite and $\Theta' \neq \Theta^*$.

	The limit of $S_i(\Theta')$ can be derived similarly,
	\vspace{-1mm}
	\begin{equation}
		\begin{aligned}
			S_i(\Theta') &\,= \,\tfrac{1}{n}P_n^{-\frac{1}{2}}\Psi_n\tr\Lambda_i(Y_n-\Phi_n\Theta') \\
			&\,=\, \tfrac{1}{n}P_n^{-\frac{1}{2}}\Psi_n\tr\Lambda_i\Phi_n(\tilde{\Theta}) + \tfrac{1}{n}P_n^{-\frac{1}{2}}\Psi_n\tr\Lambda_i W_n,
		\end{aligned}
	\end{equation}
	We will examine the asymptotic behavior of the two terms separately again. The convergence of $\tfrac{1}{n}P_n^{-\frac{1}{2}}\Psi_n\tr\Lambda_i W_n$ follows from \eqref{equ:CS-Kolmogorov}, since the variance of $\alpha_{i,k}\psi_{k,j}w_{k,l}$ equals the variance of $\psi_{k,j}w_{k,l}$, therefore
	\vspace{-1mm}
	\begin{equation}
		\tfrac{1}{n}P_n^{-\frac{1}{2}}\Psi_n\tr\Lambda_i W_n \xrightarrow{a.s.} 0,\quad \text{ as } \quad n \to \infty.
	\end{equation}
	In the first term, it holds that, $\{P_n^{-\frac{1}{2}}\} \xrightarrow{a.s.} \{P^{-\frac{1}{2}}\}$ and $\tilde{\Theta}$ is constant, thus it is enough to prove that $\tfrac{1}{n}\Psi_n\tr\Lambda_i\Phi_n$ converges almost surely to 0 element-wise. First, we fix a realization of every random variable except the random signs ($\Lambda_i$). For this realization the assumptions \ref{assu:r_grr} and \ref{assu:i_grr} hold. Then, $\{\alpha_{i,k}\psi_{k,j}\varphi_{k,l}\}$ becomes a sequence of conditionally independent random variables with conditional covariances $(\psi_{k,j}\varphi_{k,l})^2$. Using \ref{assu:r_grr} and \ref{assu:i_grr}, we get
	\vspace{-1mm}
	\begin{equation}
		\begin{aligned}
			&\sum_{k=0}^{\infty}\frac{(\psi_{k,j}\varphi_{k,l})^2}{k^2} \leq
			\sum_{k=0}^{\infty}\frac{\norm{\psi_{k}}^2}{k}\frac{\norm{\varphi_{k}}^2}{k} \leq \\
			&\sqrt{\sum_{k=0}^{\infty}\frac{\norm{\psi_{k}}^4}{k^2}}\sqrt{\sum_{k=0}^{\infty}\frac{\norm{\varphi_{k}}^4}{k^2}} < \infty, 
		\end{aligned}
	\end{equation}
	therefore, by applying Kolmogorov’s SLLN element-wise,
	\begin{equation}
		\tfrac{1}{n}P_n^{-\frac{1}{2}}\Psi_n\tr\Lambda_i\Phi_n\tilde{\Theta}\, \xrightarrow{a.s.}\, 0,
	\end{equation}
	as {\newtext $n \to \infty$} holds for (almost) any realization, therefore holds true unconditionally. From the previous derivation
	\vspace{-1mm}
	\begin{equation}
		\begin{aligned}
			&\norm{S_i(\Theta')}\frob^2 \,\xrightarrow{a.s.}\,	0,
		\end{aligned}
		\vspace{-1mm}
	\end{equation}
	as {\newtext $n \to \infty$}, for each $i \in \{1,\dots,m-1\}$.

	Previous results showed that for each $i$,
	the function $\norm{S_i(\Theta')}\frob^2$ converges	with probability 1. 
    As a consequence, {\newtext for each realization $\omega \in \Omega$ (from an event with probability one, where $(\Omega, \mathcal{F}, \mathbb{P})$ is the underlying probability space),}
    and for each $\delta >0$, there exists a (realization dependent) $N(\omega)>0$ such that for $n \geq N$ (from now on, $i \neq 0$),
	\vspace{-1mm}
	\begin{align}
		\norm{P_n^{-\frac{1}{2}}V_n - P^{-\frac{1}{2}}V}\frob \leq \delta,
		\hspace{2mm}\norm{\tfrac{1}{n}P_n^{-\frac{1}{2}}\Psi_n\tr W_n}\frob \leq \delta,\\[2mm]
		\norm{\tfrac{1}{n}P_n^{-\frac{1}{2}}\Psi_n\tr\Lambda_i\Phi_n}\frob \leq \delta,
		\hspace{2mm}\norm{\tfrac{1}{n}P_n^{-\frac{1}{2}}\Psi_n\tr\Lambda_iW_n}\frob \leq \delta.
	\end{align}
	{\newtext Then}, for all $n \geq N$, we have
	\vspace{-1mm}
		\begin{align}
			&\norm{S_0(\Theta')}\frob = \norm{P_n^{-\frac{1}{2}}V_n\tilde{\Theta} + \tfrac{1}{n}P_n^{-\frac{1}{2}}\Psi_n\tr W_n}\frob = \notag\\
			& \norm{(P_n^{-\frac{1}{2}}V_n - P^{-\frac{1}{2}}V)\tilde{\Theta} + P^{-\frac{1}{2}}V\tilde{\Theta} +  \tfrac{1}{n}P_n^{-\frac{1}{2}}\Psi_n\tr W_n}\frob = \notag\\
			&\norm{- P^{-\frac{1}{2}}V\tilde{\Theta} -(P_n^{-\frac{1}{2}}V_n - P^{-\frac{1}{2}}V)\tilde{\Theta} - \tfrac{1}{n}P_n^{-\frac{1}{2}}\Psi_n\tr W_n}\frob \geq \notag\\
			&\norm{- P^{-\frac{1}{2}}V\tilde{\Theta}}\frob -\norm{(P_n^{-\frac{1}{2}}V_n - P^{-\frac{1}{2}}V)\tilde{\Theta}}\frob \notag\\ 
			& -  \norm{\tfrac{1}{n}P_n^{-\frac{1}{2}}\Psi_n\tr W_n}\frob \geq
			\norm{P^{-\frac{1}{2}}V\tilde{\Theta}}\frob -\delta\norm{\tilde{\Theta}}\frob-  \delta =\notag\\
			&\norm{U_{\sigma}\Sigma V_{\sigma}\tr\tilde{\Theta}}\frob -\delta\norm{\tilde{\Theta}}\frob-  \delta \geq \notag\\
			&\sigma_{\text{min}}(P^{-\frac{1}{2}}V)\norm{\tilde{\Theta}}\frob -\delta\norm{\tilde{\Theta}}\frob-  \delta,\\[-7mm]
			\nonumber
		\end{align}
	where $U_{\sigma}\Sigma V_{\sigma}$ is the SVD decomposition of $P^{-\frac{1}{2}}V$ and $\sigma_{\text{min}}(\cdot)$ denotes the smallest singular value. We also have
	\vspace{-1mm}
		\begin{align}
			&\norm{S_i(\Theta')}\frob =  \norm{\tfrac{1}{n}P_n^{-\frac{1}{2}}\Psi_n\tr\Lambda_i\Phi_n\tilde{\Theta} + \tfrac{1}{n}P_n^{-\frac{1}{2}}\Psi_n\tr\Lambda_i W_n}\frob \notag\\
			& \leq \norm{\tfrac{1}{n}P_n^{-\frac{1}{2}}\Psi_n\tr\Lambda_i\Phi_n}\frob \norm{\tilde{\Theta}}\frob + \norm{\tfrac{1}{n}P_n^{-\frac{1}{2}}\Psi_n\tr\Lambda_i W_n}\frob \notag\\
			& \leq
			\delta\norm{\tilde{\Theta}}\frob + \delta.
		\end{align}
	Hence, we have $\norm{S_i(\Theta')}\frob < \norm{S_0(\Theta')}\frob, \forall\, \Theta'$ that satisfy
	\vspace{-1mm}
	\begin{equation}
		\delta\norm{\tilde{\Theta}}\frob + \delta < \sigma_{\text{min}}(P^{-\frac{1}{2}}V)\norm{\tilde{\Theta}}\frob -\delta\norm{\tilde{\Theta}}\frob-  \delta,
	\end{equation}
	which can be reformulated as
	\vspace{-1mm}	
	\begin{equation}
		\kappa_0(\delta)\, \defeq\, \frac{2\delta}{\sigma_{\text{min}}(P^{-\frac{1}{2}}V) - 2\delta} < \norm{\tilde{\Theta}}\frob\!,
	\end{equation}
	therefore, those $\Theta'$ for which $\kappa_0(\delta) < \norm{\Theta^* - \Theta'}\frob$ are not in the confidence region ${\Upsilon}_n$, for $n \geq N$. For any 
	$\varepsilon > 0$, by setting $\delta = (\varepsilon\sigma_{\text{min}}(P^{-\frac{1}{2}}V))/(2+2\varepsilon)$, we have $\Upsilon_n \subseteq \mathcal{B}_{\varepsilon}(\Theta^*)$, therefore, the claim of the theorem follows.
\end{proof}

\end{document}